\documentclass[aps,pra,reprint,superscriptaddress,longbibliography]{revtex4-2}

\usepackage{amsmath,amssymb}
\usepackage{graphicx}
\usepackage[nopatch=footnote]{microtype}
\usepackage{placeins}
\usepackage{url}
\usepackage{tikz}
\usetikzlibrary{calc,decorations.pathreplacing}
\usepackage[dvipsnames]{xcolor}
\usepackage[colorlinks=true,allcolors=blue]{hyperref}

\newcommand{\ket}[1]{\lvert #1\rangle}
\newcommand{\bra}[1]{\langle #1\rvert}
\newcommand{\braket}[2]{\langle #1\vert #2\rangle}
\newcommand{\Pusd}{P_{\mathrm{USD}}}
\newcommand{\Eenv}{\mathsf E}
\newcommand{\Rt}{\mathsf R}
\DeclareMathOperator{\tr}{tr}
\newtheorem{proposition}{Proposition}

\begin{document}

\title{Symmetry-resolved tree tensor network analysis of Bell-state discrimination with ancilla-assisted passive linear optics}

\author{Anand Kumar}
\email{anand89694@bhu.ac.in}
\affiliation{Department of Physics, Institute of Science, Banaras Hindu University, Varanasi 221005, India}

\author{Wojciech Roga}
\affiliation{Department of Electronics and Electrical Engineering, Keio University, 3-14-1 Hiyoshi, Kohoku-ku, Yokohama 223-8522, Japan}

\author{Devendra Kumar Mishra}
\affiliation{Department of Physics, Institute of Science, Banaras Hindu University, Varanasi 221005, India}

\author{Masahiro Takeoka}
\affiliation{Department of Electronics and Electrical Engineering, Keio University, 3-14-1 Hiyoshi, Kohoku-ku, Yokohama 223-8522, Japan}

\begin{abstract}
Ancillary photons allow passive linear-optical Bell-state discrimination to exceed the one-half probability of success limit achievable with vacuum auxiliary modes. For a fixed analyzer and ancillary resource, the discrimination probability is determined by the photon-counting patterns that occur uniquely for each Bell input. We develop a photon-number-resolved tree tensor network method for evaluating these detector supports in recursively structured analyzers. Photon-number conservation and mode symmetries separate most of the Bell-state outputs before the remaining detector support is examined. We establish a tree tensor network based construction to evaluate the remaining support without need of recalling all relevant photon detector patterns separately. We apply the method to recursive, product, and asymmetric ancillary states with up to $32$ optical modes and reproduce known analytical and literature benchmarks. For factorized ancillary resources with fixed photon numbers in halves of the analyzer and the symmetry property used in our analysis, we also derive an exact composition relation; for reflected asymmetric pairings the success probability is the arithmetic mean of the corresponding symmetric configurations. The complete enumeration for smaller systems, permanent-based amplitude calculations, and independent evaluations provide additional checks of the numerical results.
\end{abstract}

\maketitle

\section{Introduction}
\label{sec:intro}
Bell-state measurements (BSMs) are fundamental to quantum teleportation, entanglement swapping, and photonic quantum information processing, including fusion-based photonic architectures~\cite{BennettTeleportation1993,KLM2001,BrowneRudolph2005,KokRMP2007,BianchiMarconiBacco2026,ZukowskiEntanglementSwapping1993,Bartolucci2023FBQC}. For four equal-prior dual-rail Bell states, a fixed passive linear-optical analyzer using only vacuum auxiliary modes and ideal photon-number-resolving detector has an unambiguous success probability no larger than one-half~\cite{Lutkenhaus1999,Calsamiglia2001}. This bound applies to that restricted optical resource class and is not a general limit on Bell measurements.

Non-vacuum photonic ancillas can raise the success probability above one-half without introducing active optical elements. Grice~\cite{Grice2011} obtained a success probability of three-quarters using an entangled ancillary resource and introduced a recursive family whose success probability approaches unity. Ewert and van Loock~\cite{EwertvanLoock2014} reached the same probability with unentangled ancillary single photons and developed a different recursive extension together with elementary-block replacements. Subsequent work has considered interferometer optimization, analytical bounds, and ancillary-single-photon constructions~\cite{SmithKaplan2018,OlivoGrosshans2018,YamazakiIkutaYamamoto2023}. In particular, Yamazaki \textit{et al.}~\cite{YamazakiIkutaYamamoto2023} obtained a maximum success probability of $403/512$ with $28$ ancillary photons. Ancilla-assisted boosted BSMs have also been demonstrated experimentally and used in photonic fusion and teleportation settings~\cite{Bayerbach2023,Hauser2025BoostedBSM,DAurelio2025BoostedTeleportation}. Other resource models include active Gaussian processing~\cite{ZaidiVanLoock2013}, hybrid photon-counting and homodyne detection~\cite{Asenbeck2024HybridBSM}, encoded or multi-photon states~\cite{LeeJeong2013,LeeParkRalphJeong2015,Hilaire2023,ReissVanLoock2026}, auxiliary photonic degrees of freedom~\cite{WeiBarreiroKwiat2007,PisentiGaeblerLynn2011,LahaVanLoock2026}, and nonlinear interactions~\cite{AkinNonlinear2025}.

For a fixed analyzer and ancillary resource, unambiguous Bell discrimination is a support problem for four photonic outputs rather than only an amplitude-evaluation problem. A photon-counting record is conclusive for one Bell input only when that record is present for the target input and absent for every competing input. Individual passive-linear-optical Fock amplitudes can be evaluated from matrix permanents~\cite{Scheel2008}, which itself can be computationally challenging problem, beyond that the success probability requires identifying all detector records that occur for exactly one Bell input. Their number grows combinatorially with the number of modes and photons; for the largest case studied here, $32$ modes and $30$ photons give about $2.33\times10^{17}$ possible occupation patterns before any symmetry is applied. Structured classical methods for linear optics can exploit photon-number conservation, sparsity, graph structure, compressive representations, and tensor-network factorizations~\cite{CliffordClifford2018,roga2020classical,OhGraph2022,jacob2020franck,Oh2024,VintherKastoryano2025,Cilluffo2026}. For Bell discrimination, Yamazaki \textit{et al.}~\cite{YamazakiIkutaYamamoto2023} developed a bosonic stabilizer formalism that organizes Fock outcomes into stabilizer spaces and provides suppression conditions and stabilizer-measurement probabilities.  

The problem addressed here is to determine, for chosen interferometers and auxiliary photonic states, Bell-label-resolved unique support of detector patterns without enumerating the full detection space. We use a photon-number-resolved tree tensor network (TTN) whose balanced tree structure follows the chosen recursive optical analyzer. 

We first use two quantum-optical constraints to remove most Bell-state competition before the remaining support is examined. The photon number in each half of the analyzer restricts which Bell-conditioned outputs can occupy the same sector. In addition, the difference between photon numbers in globally even and odd indexed modes, taken modulo four, further separates the Bell labels. For the ancillary resources considered here, these constraints leave only Bell states $\phi^+$--$\phi^-$ competition in the relevant sectors. 

Photon-number conservation implies an additional structure in the virtual bonds of TTN~\cite{ShiDuanVidal2006,tagliacozzo2009simulation,Orus2014,SinghPfeiferVidal2010,SinghPfeiferVidal2011}, so each subtree is organized by the number of photons it contains. Rather than constructing the complete detector space, our calculation retains the support information passed through the analyzer hierarchy and, at the final cut, through its halves. The TTN is used here as a structured finite instance representation of the Bell-conditioned outputs. In this work, we do not claim any general asymptotic computational advantage.

The optical structure also yields an exact result independent of the TTN representation. For ancillary inputs that factorize between the two analyzer halves, with definite photon number and a fixed even--odd photon-number difference modulo four in each half, the outer $\phi^+$ and $\phi^-$ branches probe the two ancillary halves separately. Their unique contributions therefore combine directly. For a resource formed from one ancillary half and the reflected ordering of another, this gives an exact composition rule. The aggregate unambiguous-discrimination probability is the arithmetic mean of the two corresponding reflected symmetric configurations. Such an asymmetric pairing therefore cannot outperform the better symmetric endpoint within this resource class. The composition relation also expresses the two outer-sector contributions in terms of half-resource support quantities. In the calculations reported here, these quantities are evaluated using the TTN and independently checked with the involved half-circuit calculation.

We apply the construction to recursive ancillary resources, products of elementary ancillary blocks, asymmetric recursive--product compositions, and a reflection-covariant mixed-block case, with calculations extending to $32$ optical modes. Known aggregate probabilities are recovered where analytical or literature benchmarks are available, while our calculation additionally resolves the root photon-number-sector contributions for each Bell state. Complete detector-space enumeration at the smallest size, permanent based calculations of selected amplitudes, comparisons between distinct TTN update paths, and separately implemented calculations on halves of the circuits provide checks on the reported results.

The remainder of the paper is organized as follows. In Sec.~\ref{sec:model}, we define the physical model, ancillary resources, analyzer, and photon-counting discrimination objective. In Sec.~\ref{sec:ttn}, we develop the symmetry-resolved TTN representation and the optical symmetry reductions, while in Sec.~\ref{sec:screen}, we introduce the so-called directional support construction to recognize support unique for a Bell state and prove its completeness. In Sec.~\ref{sec:results}, we present the support-resolved results including probabilities for unambiguous discrimination of each and discuss numerical validation. In Sec.~\ref{sec:discussion} we discuss the scope, limitations, and outlook of the method.

\section{Physical model and discrimination objective}\label{sec:model}

We consider a fixed passive Bell-state analyzer supplied with ancillary photonic states and followed by photon-number-resolving detectors (PNRDs).

\subsection{Bell inputs and ancillary resources}

We first specify how the Bell input and ancillary photons are arranged at the analyzer input. We consider $M=2^{N+2}$ optical modes, with $N\geq1$. The four modes carrying the dual-rail Bell state lie at the center, while the ancillary modes occupy the two sides. The four Bell modes, in increasing physical indices, are
\[(b_0,b_1,b_2,b_3)=\left(\frac{M}{2}-2,\frac{M}{2}-1,\frac{M}{2},\frac{M}{2}+1\right).\] Modes $(b_0,b_1)$ and $(b_2,b_3)$ form the two dual-rail pairs. We refer to modes $0,\ldots,M/2-1$ and $M/2,\ldots,M-1$ as halves $A$ and $B$ of the analyzer, respectively, so the cut $A\mid B$ lies between the two pairs of Bell modes.

In the ordered modes $(b_0,b_1,b_2,b_3)$, the four Bell inputs are
\begin{align}\ket{\mathcal B_{\psi^\pm}}&=\frac{1}{\sqrt2}\left(\ket{1,0,0,1}\pm\ket{0,1,1,0}\right),\nonumber\\
\ket{\mathcal B_{\phi^\pm}}&=\frac{1}{\sqrt2}\left(\ket{1,0,1,0}\pm\ket{0,1,0,1}\right).\label{eq:bell_states}\end{align}
We use $\ell\in\{\psi^+,\psi^-,\phi^+,\phi^-\}$ to label the Bell states. The ancillary states $\ket{\mathcal A_A}$ and $\ket{\mathcal A_B}$ occupy the ancillary modes of the halves $A$ and $B$, respectively, so the complete input is
\begin{equation}
\ket{\Psi_\ell^{\mathrm{in}}}
=\ket{\mathcal A_A}\otimes\ket{\mathcal B_\ell}\otimes\ket{\mathcal A_B}.
\label{eq:input_bell_ancilla}
\end{equation}
The tensor product in Eq.~\eqref{eq:input_bell_ancilla} follows the ordering of the physical modes. It is not a factorization across the $A\mid B$ cut because the four-mode Bell state spans across both halves of the analyzer.

We use the ancillary states introduced in Ewert and van Loock~\cite{EwertvanLoock2014} defined as
\[\ket{\gamma_\tau}=\frac{\ket{2,0}^{\otimes 2^{\tau-1}}+\ket{0,2}^{\otimes2^{\tau-1}}}{\sqrt2},\quad \tau\geq1.\]
In particular, \[\ket{\gamma_1}=\frac{\ket{2,0}+\ket{0,2}}{\sqrt2}.\]
The tensor factors follow increasing physical-mode order. The block $\ket{\gamma_\tau}$ occupies $2^\tau$ modes and contains $2^\tau$ photons.

In what we refer to this as the recursive resource, the two ancillary state blocks are
\[\begin{aligned}\ket{\mathcal A_A}&=\ket{\gamma_N}\otimes\ket{\gamma_{N-1}}\otimes\cdots\otimes\ket{\gamma_1},\\ 
\ket{\mathcal A_B} &=\ket{\gamma_1}\otimes\ket{\gamma_2}\otimes\cdots\otimes\ket{\gamma_N}.\end{aligned}\] The factors are mirrored across the Bell modes, so $\ket{\gamma_1}$ is adjacent to the Bell input on both sides. 

As a second resource family, we use only elementary $\gamma_1$ blocks,
\[\ket{\mathcal A_A}=\ket{\mathcal A_B}=\ket{\gamma_1}^{\otimes(2^N-1)}.\]
In both resource families, each ancillary half occupies $M/2-2$ modes and contains $M/2-2$ photons. The two resource families coincide at $M=8$. For $M\geq16$, the recursive resource contains the higher order blocks $\gamma_2,\ldots,\gamma_N$, whereas the product resource contains only the blocks $\gamma_1$.

\subsection{Fixed passive analyzer}
\label{subsec:fixed-analyzer}
We use the fixed passive analyzer of Ewert and van Loock~\cite{EwertvanLoock2014}, with the physical-mode ordering defined above and the beam-splitter convention given below. It consists of a central Bell layer, followed by two identical transformations that act independently on the analyzer halves $A$ and $B$. We understand the output modes as those corresponding to the rows of the matrix and input modes as those corresponding to the columns.

We adapt the convention for the balanced beam-splitter as follows
\[U_{\mathrm{BS}}=\frac{1}{\sqrt2}\begin{pmatrix}1&\mathrm{i}\\ \mathrm{i}&1\end{pmatrix}.
\] With this, a balanced beam splitter maps $\ket{1,1}$ to $\mathrm{i}\ket{\gamma_1}$.

Within each half, the modes are arranged in increasing order of physical index. Each half of the analyzer contains $2^{N+1}=M/2$ modes. At order $N$, two copies of the order-$(N-1)$ transformation act on equal subblocks, and balanced beam splitters mix the corresponding modes. Using $\mathbb I_2$ for the $2\times2$ identity operator, the mode transformation on either half is
\begin{align}U_{\mathrm h}^{(N)}&=\frac{1}{\sqrt2}\begin{pmatrix}U_{\mathrm h}^{(N-1)}&\mathrm{i}U_{\mathrm h}^{(N-1)}\\ \mathrm{i}U_{\mathrm h}^{(N-1)}&U_{\mathrm h}^{(N-1)}
\end{pmatrix}\nonumber\\ &=U_{\mathrm{BS}}\otimes U_{\mathrm h}^{(N-1)}
=U_{\mathrm{BS}}^{\otimes N}\otimes\mathbb I_2,\quad U_{\mathrm h}^{(0)}=\mathbb I_2.
\label{eq:half_analyzer}\end{align}

The central Bell layer mixes the mode pairs $(b_0,b_2)$ and $(b_1,b_3)$. Here $(U_{\mathrm{BS}})_{ij}$ denotes $U_{\mathrm{BS}}$ acting on modes $i,j$ and the identity on all remaining modes, and $\oplus$ denotes the block-diagonal direct sum. The Bell layer and the complete single-particle analyzer are
\begin{align}U_{\mathrm{Bell}}&=(U_{\mathrm{BS}})_{b_0b_2}(U_{\mathrm{BS}})_{b_1b_3},\nonumber\\ U_{\mathrm{an}}&=\left(U_{\mathrm h}^{(N)}\oplus U_{\mathrm h}^{(N)}\right)U_{\mathrm{Bell}}.\label{eq:full_analyzer}\end{align}
The rightmost factor acts first. Thus, the Bell layer couples the two halves of the analyzer, after which the operators there act independently. The network contains only passive beam splitters and therefore conserves total photon number. The interferometer is fixed throughout; we do not optimize over other circuits. The $M=8$ instance is shown in Fig.~\ref{fig:analyzer-decomposition}.

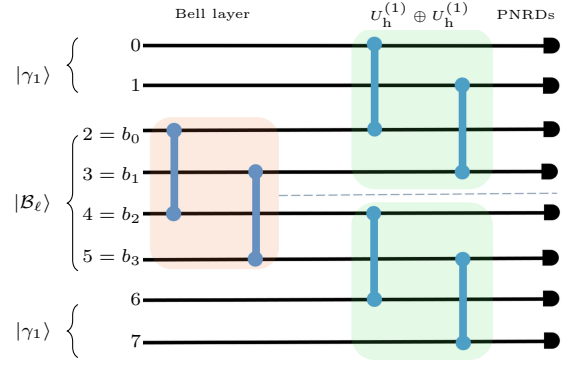
\begin{figure}[t]
\centering
\begin{tikzpicture}
\node[inner sep=0] (img) {\includegraphics[height=0.185\textheight]{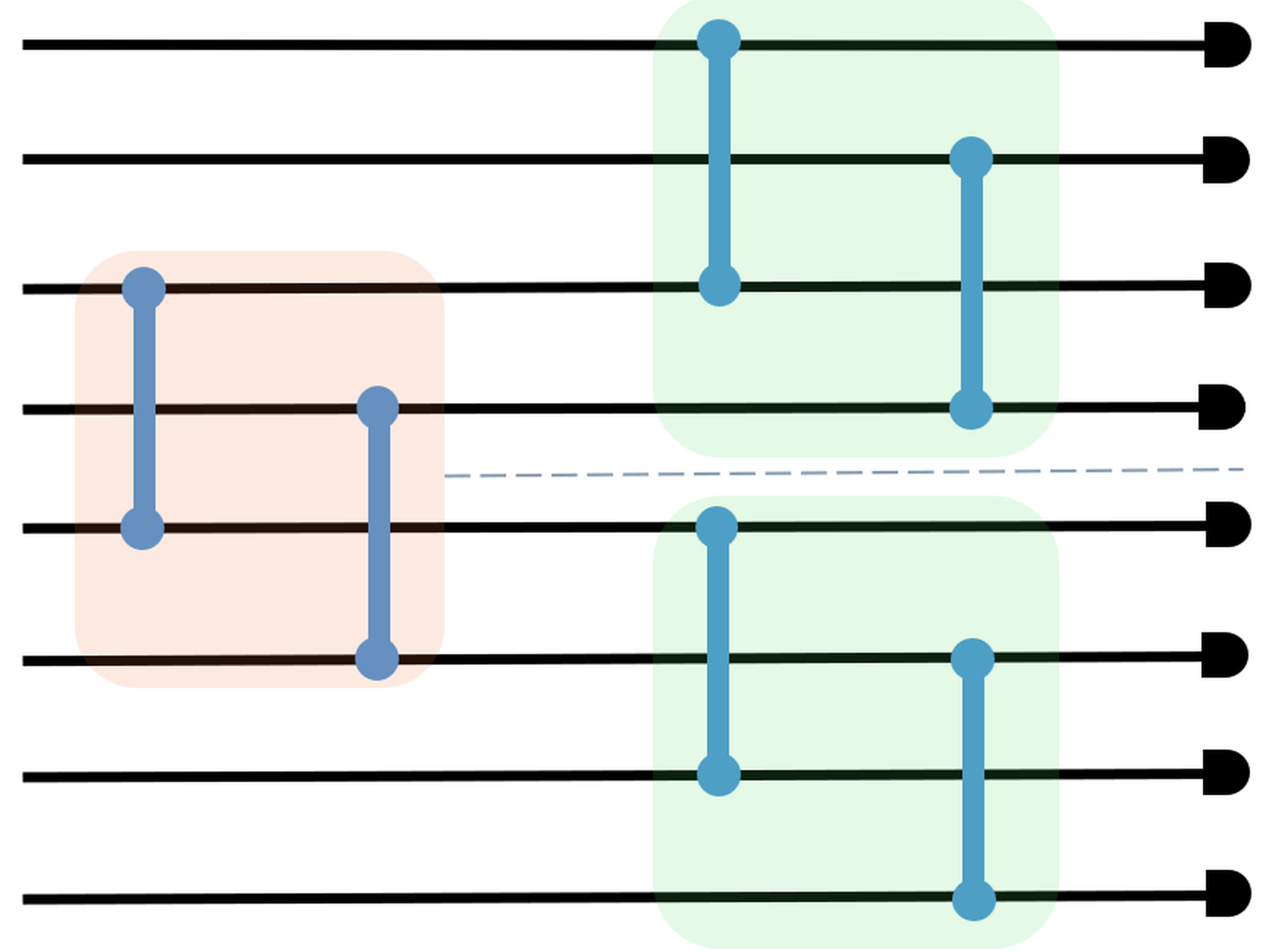}};
\begin{scope}[shift={(img.south west)},x={($(img.south east)-(img.south west)$)},y={($(img.north west)-(img.south west)$)}]
\node[font=\tiny] at (0.200,1.045) {Bell layer};
\node[font=\tiny] at (0.670,1.045) {$U_{\mathrm h}^{(1)}\oplus U_{\mathrm h}^{(1)}$};
\node[font=\tiny] at (0.910,1.045) {PNRDs};
\draw[decorate,decoration={brace,amplitude=4pt,mirror},thin](-0.105,0.975) -- (-0.105,0.810);
\draw[decorate,decoration={brace,amplitude=4pt,mirror},thin](-0.105,0.680) -- (-0.105,0.272);
\draw[decorate,decoration={brace,amplitude=4pt,mirror},thin](-0.105,0.170) -- (-0.105,0.010);
\node[font=\scriptsize,anchor=east] at (-0.145,0.862) {$\ket{\gamma_1}$};
\node[font=\scriptsize,anchor=east] at (-0.145,0.476) {$\ket{\mathcal B_\ell}$};
\node[font=\scriptsize,anchor=east] at (-0.145,0.090) {$\ket{\gamma_1}$};
\foreach \yy/\lab in {0.953/$0$,0.833/$1$,0.682/$2=b_0$,0.560/$3=b_1$,0.441/$4=b_2$,0.313/$5=b_3$,0.187/$6$,0.061/$7$}{
\node[font=\scriptsize,anchor=east,fill=white,inner sep=0.6pt] at (0.043,\yy) {\lab};
}
\end{scope}
\end{tikzpicture}
\caption{Optical analyzer for the $M=8$ case with two ancillary $\ket{\gamma_1}$ blocks. The Bell layer $U_{\mathrm{Bell}}$ (orange box) is followed by the two half analyzers $U_{\mathrm h}^{(1)}$ (green boxes) and photon-number-resolving detectors. The dashed line marks the global $A\mid B$ cut.}
\label{fig:analyzer-decomposition}
\end{figure}

For a mode transformation $U$, let $\Gamma(U)$ denote the corresponding number-preserving transformation in the Fock space. Suppressing the ancillary factors and writing only the two Bell modes in each half, the Bell layer gives
\begin{align}
\Gamma(U_{\mathrm{Bell}})\ket{\mathcal B_{\psi^+}}
&=\frac{\mathrm{i}}{\sqrt2}\left(\ket{1,1}_A\ket{0,0}_B+\ket{0,0}_A\ket{1,1}_B\right),\nonumber\\
\Gamma(U_{\mathrm{Bell}})\ket{\mathcal B_{\psi^-}}
&=\frac{1}{\sqrt2}\left(\ket{1,0}_A\ket{0,1}_B-\ket{0,1}_A\ket{1,0}_B\right),\nonumber\\
\Gamma(U_{\mathrm{Bell}})\ket{\mathcal B_{\phi^\pm}}
&=\frac{\mathrm{i}}{2}\left(\ket{2,0}\pm\ket{0,2}\right)_A\ket{0,0}_B\nonumber\\[-1mm]
&\quad+\frac{\mathrm{i}}{2}\ket{0,0}_A\left(\ket{2,0}\pm\ket{0,2}\right)_B.
\label{eq:bell_layer_action}
\end{align}
The $\psi^-$ Bell-layer output has one photon in each half of the analyzer. For $\psi^+$ and $\phi^\pm$, both photons occupy half $A$ or half $B$.

\subsection{Photon-counting discrimination objective}
Including the two photons from the Bell states, the input therefore contains
\[Q=M-2\]  photons in total. A PNRD record is an occupation vector
\[\mathbf n=(n_0,\ldots,n_{M-1}),\] where $n_j$ is the nonnegative integer photon number detected in output mode $j$. The allowed fixed-$Q$ records form the set
\[\Omega_{M,Q}=\left\{\mathbf n:\sum_{j=0}^{M-1}n_j=Q\right\}.\]

For Bell label $\ell$, the output state and the amplitude of record $\mathbf n$ are
\begin{align*}
\ket{\Psi_\ell^{\mathrm{out}}}
&=\Gamma(U_{\mathrm{an}})\ket{\Psi_\ell^{\mathrm{in}}}
=\sum_{\mathbf n\in\Omega_{M,Q}}c_\ell(\mathbf n)\ket{\mathbf n},\\
c_\ell(\mathbf n)
&=\braket{\mathbf n}{\Psi_\ell^{\mathrm{out}}}
=\bra{\mathbf n}\Gamma(U_{\mathrm{an}})\ket{\Psi_\ell^{\mathrm{in}}}.
\end{align*}
The corresponding detection probability is $|c_\ell(\mathbf n)|^2$.

Unambiguous state discrimination (USD) allows for an inconclusive outcome but excludes an incorrect conclusive assignment~\cite{Ivanovic1987,Dieks1988,Peres1988,Chefles1998}. The Bell states are mutually orthogonal, so here, the inconclusive outcome arises from the restriction to the class of measurements rather than from non-orthogonality of the inputs. Let $\mu_\ell$ denote the probability of detector records that occur for Bell input $\ell$ and for no other Bell state.
\begin{equation}
\mu_\ell
=\sum_{\substack{\mathbf n\in\Omega_{M,Q}\\
c_\ell(\mathbf n)\neq0,\; c_k(\mathbf n)=0\;\forall\,k\neq\ell}}
|c_\ell(\mathbf n)|^2,
\label{eq:unique_mass}
\end{equation}
where $k$ runs over the other Bell labels. For equal Bell priors,
\begin{equation}
\Pusd
=\frac14\sum_\ell \mu_\ell,
\label{eq:usd_success}
\end{equation}
where the sum runs over the four Bell labels. For this fixed PNRD measurement, assigning every uniquely supported record to its Bell label gives the optimal unambiguous assignment. Indeed, any conclusive assignment to $\ell$ must use a record for which all competing amplitudes vanish, so including every such record maximizes the conclusive probability for the fixed measurement.

The zeros in Eq.~\eqref{eq:unique_mass} refer to exact amplitudes. Throughout this work, we assume indistinguishable photons, lossless passive optics, ideal ancillary preparation, and ideal photon-number-resolving detectors.

\section{TTN representation and symmetry reduction}\label{sec:ttn}

We represent each Bell-labeled output state by a balanced, photon-number-resolved TTN. The tree follows the recursive structure of the fixed analyzer, while photon number is carried explicitly on every virtual bond. The two branches meeting at the global root correspond to the physical analyzer halves $A$ and $B$. Figure~\ref{fig:ttn-us-psi} illustrates the representation; the TTN is a representation of the optical state, not an additional optical network.

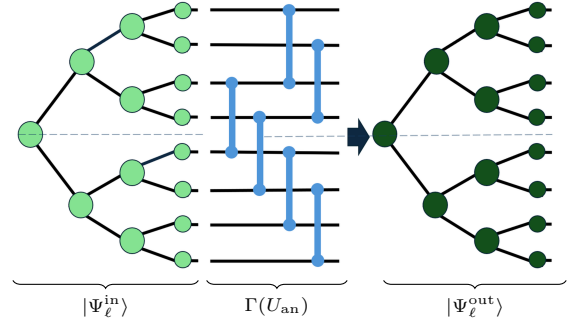
\begin{figure}[t]
\centering
\begin{tikzpicture}
\node[inner sep=0] (img) {\includegraphics[height=0.160\textheight]{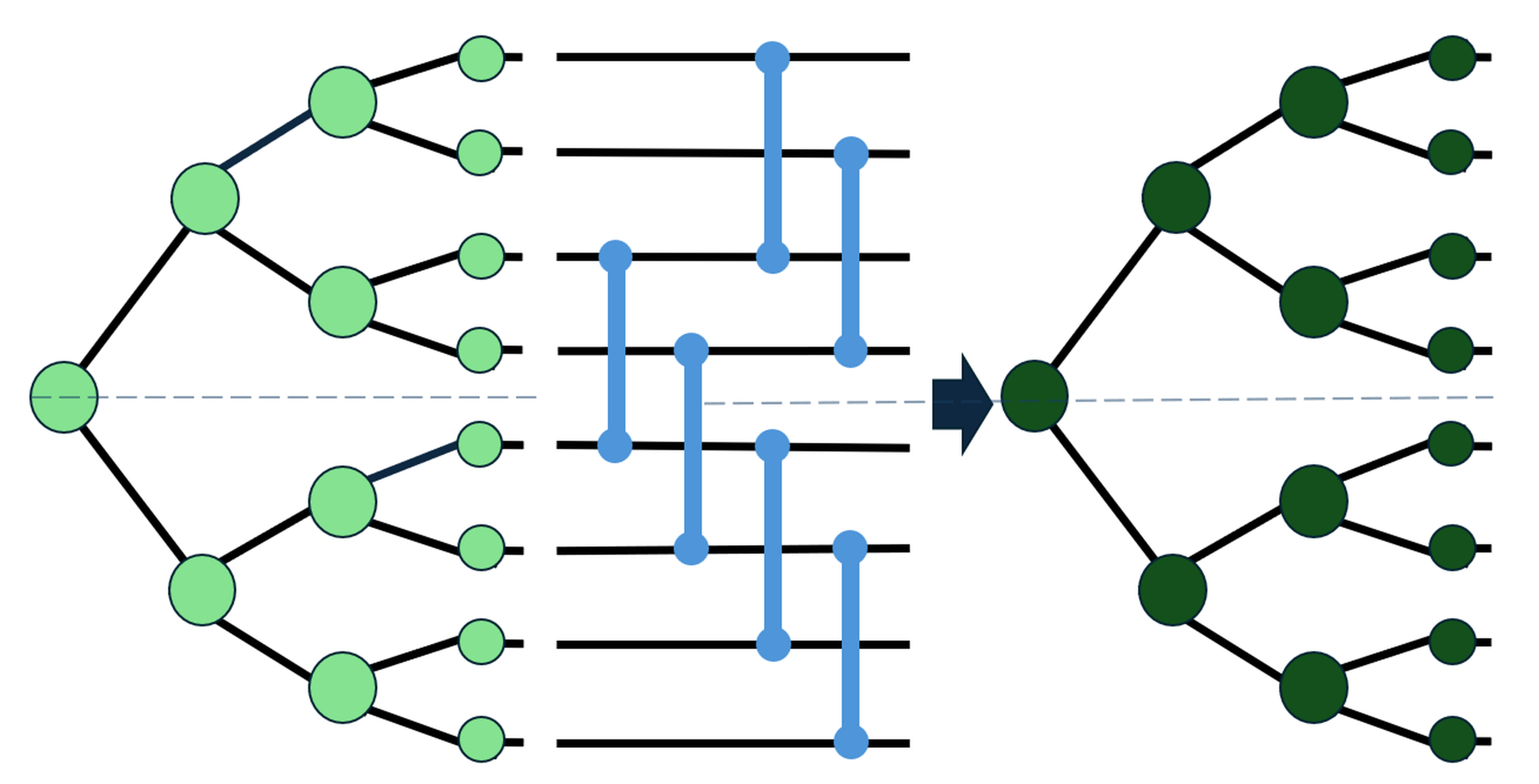}};
\begin{scope}[shift={(img.south west)},x={($(img.south east)-(img.south west)$)},y={($(img.north west)-(img.south west)$)}]
\draw[decorate,decoration={brace,mirror,amplitude=3pt},thin]
(0.010,-0.015) -- (0.340,-0.015);
\node[font=\scriptsize] at (0.175,-0.105)
{$\ket{\Psi_\ell^{\mathrm{in}}}$};
\draw[decorate,decoration={brace,mirror,amplitude=3pt},thin]
(0.355,-0.015) -- (0.605,-0.015);
\node[font=\scriptsize] at (0.480,-0.105)
{$\Gamma(U_{\mathrm{an}})$};
\draw[decorate,decoration={brace,mirror,amplitude=3pt},thin]
(0.670,-0.015) -- (0.995,-0.015);
\node[font=\scriptsize] at (0.833,-0.105)
{$\ket{\Psi_\ell^{\mathrm{out}}}$};
\end{scope}
\end{tikzpicture}
\caption{TTN representation of the optical-state evolution through the analyzer. The input state $\ket{\Psi_\ell^{\mathrm{in}}}$ is evolved by $\Gamma(U_{\mathrm{an}})$ to the output state $\ket{\Psi_\ell^{\mathrm{out}}}$. Each leaf of the input and output TTNs corresponds to one optical mode, and the dashed line denotes the global $A\mid B$ cut.}
\label{fig:ttn-us-psi}
\end{figure}

\subsection{Photon-number-resolved tree representation}
\label{subsec:ttn-boundary}

For $M=2^{N+2}$, we use a balanced binary tree with one optical mode at each leaf~\cite{ShiDuanVidal2006,Seitz2023simulatingquantum}. For a non-root internal node $v$, its ordered children are denoted by $v_{\mathrm L}$ and $v_{\mathrm R}$. These are local tree labels, distinct from the global analyzer halves $A$ and $B$.

We use the $U(1)$ symmetry associated with photon-number conservation. Each virtual bond is resolved into sectors labeled by the photon number $q$ in the physical modes below that bond~\cite{SinghPfeiferVidal2010,SinghPfeiferVidal2011}. Since the four Bell-conditioned states are represented and evolved separately, the dimension of a given photon-number block can depend on the Bell label. We denote \(r_{\ell,v}^{(q)}\) as the dimension of the represented $q$-photon block on the bond leaving a non-root node $v$.

For the detector occupation vector $\mathbf n$ defined in Sec.~\ref{sec:model}, $\mathbf n_v$ denotes its restriction, in leaf order, to the physical modes in the subtree rooted at $v$. The subtree records containing $q$ photons form
\[
\Omega_{v,q}
=
\left\{
\mathbf n_v:
\sum_{j\in v}n_j=q
\right\},
\]
where $j\in v$ means that physical mode $j$ is a leaf below node $v$.

For a non-root internal node $v$, write $q_v$ for the photon number in its subtree and $q_{\mathrm L}$ and $q_{\mathrm R}$ for the photon numbers in its two child subtrees. Photon-number conservation requires $q_v=q_{\mathrm L}+q_{\mathrm R}$. At fixed compatible photon numbers, the tensor at $v$ maps the two child blocks into the parent block. We label their basis indices by $\alpha=1,\ldots,r_{\ell,v_{\mathrm L}}^{(q_{\mathrm L})}$, $\beta=1,\ldots,r_{\ell,v_{\mathrm R}}^{(q_{\mathrm R})}$, and $\omega=1,\ldots,r_{\ell,v}^{(q_v)}$. Let $\mathsf T_v^{(\ell)}$ denote the TTN tensor at node $v$, with components
\[
\bigl[\mathsf T_v^{(\ell)}\bigr]^{(q_v,\omega)}
_{(q_{\mathrm L},\alpha)(q_{\mathrm R},\beta)}.
\]
At the lowest internal level the children are physical modes $i$ and $j$, so $q_v=n_i+n_j$.

For a fixed subtree record $\mathbf n_v\in\Omega_{v,q_v}$, contracting all tensors and physical occupations below $v$ while leaving the outgoing virtual bond open gives the boundary vector
\[
\mathbf f_{\ell,v}^{(q_v)}(\mathbf n_v)
\in
\mathbb C^{r_{\ell,v}^{(q_v)}}.
\]
It contains the amplitude information from that subtree record passed to the remainder of the TTN. Its component along basis state $\omega$ is denoted by $f_{\ell,v}^{(q_v,\omega)}(\mathbf n_v)$.

At the lowest internal level, where $v$ joins physical modes $i$ and $j$,
\[
f_{\ell,v}^{(q_v,\omega)}(n_i,n_j)
=
\bigl[\mathsf T_v^{(\ell)}\bigr]^{(q_v,\omega)}_{n_i n_j},
\qquad q_v=n_i+n_j .
\]
At a higher non-root node, a specified detector record fixes $q_{\mathrm L}=\sum_{j\in v_{\mathrm L}}n_j$ and $q_{\mathrm R}=\sum_{j\in v_{\mathrm R}}n_j$. The boundary vector is then
\begin{align}
f_{\ell,v}^{(q_v,\omega)}(\mathbf n_v)
={}&
\sum_{\alpha,\beta}
\bigl[\mathsf T_v^{(\ell)}\bigr]^{(q_v,\omega)}
_{(q_{\mathrm L},\alpha)(q_{\mathrm R},\beta)}
\nonumber\\
&\times
f_{\ell,v_{\mathrm L}}^{(q_{\mathrm L},\alpha)}(\mathbf n_{v_{\mathrm L}})
 f_{\ell,v_{\mathrm R}}^{(q_{\mathrm R},\beta)}(\mathbf n_{v_{\mathrm R}}).
\label{eq:internal_boundary_amplitude}
\end{align}
Applied to the evolved output TTN, this recursion and the root contraction of Sec.~\ref{subsec:root-sectors} give any requested coefficient $c_\ell(\mathbf n)$ without constructing the complete output vector.

\subsection{Symmetry-preserving optical evolution}

For the resources defined in Sec.~\ref{sec:model}, each TTN leaf uses a local Fock basis containing every occupation reachable during the analyzer evolution. Each ancillary half contains $M/2-2$ photons, and the Bell layer can place at most two additional photons in either half. Thus no analyzer half contains more than $M/2$ photons. The subsequent half analyzer is passive and preserves this bound, so a complete local basis for the physically reachable occupations are
\[
\left\{\ket0,\ket1,\ldots,\ket{d_{\mathrm{loc}}-1}\right\},
\qquad
d_{\mathrm{loc}}=\frac M2+1 .
\]
For $M=8,16,32$ this gives $d_{\mathrm{loc}}=5,9,17$, respectively. Thus the local basis is not a truncation of the physically accessible Fock space in the calculations reported here.

Each balanced beam splitter preserves the total photon number of the two modes it acts on. For every physically reachable pair photon number, the local basis therefore contains the complete fixed-number block, and the corresponding beam-splitter block is unitary. Gates acting on sibling modes are absorbed directly into the associated tensor. For nonsibling modes, the physical legs are routed through the tree, the same two-mode optical gate is applied, and the affected path is refactorized independently within each photon-number sector. In exact arithmetic, retaining every nonzero singular value changes only the TTN representation and leaves the optical state unchanged. Each sectorwise refactorization uses a one-sided singular-value decomposition (SVD) gauge. The singular values are absorbed into a single factor, and the other factor is kept isometric; the orientation follows the local routed split. The evolved TTN is therefore not assumed to be globally root-canonical. After the analyzer has been applied, $\mathsf T_v^{(\ell)}$ denotes the tensors of the evolved output TTN used below.

\subsection{Photon-number-resolved environments and root sectors}
\label{subsec:root-sectors}

In this section, we describe the objects called environments. They play a crucial role in finding probabilities of each photon number sector for each Bell state, as well as the unique probabilities in each sectors as will be described in the subsequent sections. The key observation is that the higher order environments can be constructed from the lower order ones recursively through the TTN, so the above mentioned probabilities can be found by means of environments without recalling all photon detection patterns separately.

The record-resolved boundary vectors can be combined into a photon-number-resolved Gram environment. For a non-root node $v$, Bell label $\ell$, and subtree photon number $q$, define
\begin{equation}
\Eenv_{\ell,v}(q)
=
\sum_{\mathbf n_v\in\Omega_{v,q}}
\mathbf f_{\ell,v}^{(q)}(\mathbf n_v)
\mathbf f_{\ell,v}^{(q)}(\mathbf n_v)^\dagger.
\label{eq:general_environment_u1}
\end{equation}
It is a $r_{\ell,v}^{(q)}\times r_{\ell,v}^{(q)}$ Hermitian positive-semidefinite matrix whose dimension is fixed by the represented bond block, not by the number of records in $\Omega_{v,q}$.
Its rank is
\begin{equation}
\operatorname{rank}\Eenv_{\ell,v}(q)
=
\dim\operatorname{span}
\left\{
\mathbf f_{\ell,v}^{(q)}(\mathbf n_v):
\mathbf n_v\in\Omega_{v,q}
\right\},
\label{eq:environment_rank_span}
\end{equation}
so a rank-one environment means that all nonzero boundary vectors in that photon-number sector lie along one virtual direction. Since the environment is positive semidefinite, a sector is absent exactly when $\tr\Eenv_{\ell,v}(q)=0$ (where $\tr$ denotes the matrix trace). These environments propagate recursively through the tree by the corresponding photon-number-resolved double-layer contraction. The explicit gate update, sectorwise refactorization, and environment recursion are given in Appendix~\ref{app:ttn_updates}.

At the global root, write $\mathbf n=(\mathbf n_A,\mathbf n_B)$ for the restrictions of the detector record to the two physical halves. Their photon numbers are $q_A=\sum_{j=0}^{M/2-1}n_j$ and $q_B=\sum_{j=M/2}^{M-1}n_j$, with $q_A+q_B=Q$. Let $Q_X$ denote the fixed ancillary photon number in resource $\ket{\mathcal A_X}$ for $X\in\{A,B\}$; thus $Q=Q_A+Q_B+2$. From the Bell-layer action in Eq.~\eqref{eq:bell_layer_action}, followed by independent number-preserving half analyzers, the only possible root photon numbers are
\begin{equation}
q_A\in\{Q_A,Q_A+1,Q_A+2\},
\qquad
q_B=Q-q_A .
\label{eq:root_charge_sectors}
\end{equation}
The central sector $q_A=Q_A+1$ is populated only by $\psi^-$, for which one Bell photon enters each half. The two outer sectors contain the branches of $\psi^+$ and $\phi^\pm$, in which both Bell state photons occupy the same half of the analyzer.

Let $v_A$ and $v_B$ denote the two child nodes entering the global root. For $X\in\{A,B\}$, use the shorthand
\[
\mathbf f_{\ell,X}^{(q)}\equiv\mathbf f_{\ell,v_X}^{(q)},
\qquad
\Eenv_{\ell,X}(q)\equiv\Eenv_{\ell,v_X}(q),
\]
\[
\Omega_{X,q}\equiv\Omega_{v_X,q},
\qquad
r_{\ell,X}^{(q)}\equiv r_{\ell,v_X}^{(q)}.
\]
Let $\mathsf T_{\mathrm{root}}^{(\ell)}$ denote the global-root tensor, with its one-dimensional upward leg suppressed. In root sector $q_A$, define the coupling matrix
\begin{equation}
[\Rt_\ell(q_A)]_{\alpha\beta}
=
[\mathsf T_{\mathrm{root}}^{(\ell)}]
_{(q_A,\alpha)(q_B,\beta)}.
\label{eq:root_coupling}
\end{equation}
It is a $r_{\ell,A}^{(q_A)}\times r_{\ell,B}^{(q_B)}$ matrix. The Fock amplitude of a detector record in this sector is
\begin{equation}
c_\ell(\mathbf n_A,\mathbf n_B)
=
\mathbf f_{\ell,A}^{(q_A)}(\mathbf n_A)^T
\Rt_\ell(q_A)
\mathbf f_{\ell,B}^{(q_B)}(\mathbf n_B).
\label{eq:root_amplitude_bilinear}
\end{equation}
The transpose is intentional. The two boundary vectors contain ket coefficients, and the root tensor contracts them bilinearly without complex conjugation.

Summing $|c_\ell|^2$ over all detector records in a fixed root sector gives
\begin{equation}
w_\ell(q_A)
=
\tr\!\left[
\Eenv_{\ell,A}(q_A)^T
\Rt_\ell(q_A)
\Eenv_{\ell,B}(q_B)
\Rt_\ell(q_A)^\dagger
\right].
\label{eq:sector_mass_trace}
\end{equation}
This is the sum of $|c_\ell|^2$ over the sector, so $w_\ell(q_A)\ge0$. Since the root photon-number sectors partition the normalized fixed-$Q$ output state,
\begin{equation}
\sum_{q_A=0}^{Q}w_\ell(q_A)=1 .
\label{eq:root_sector_normalization}
\end{equation}

Nonzero Gram environments in both halves do not, by themselves, guarantee a nonzero global sector, because the root contraction can still vanish on their supports. We therefore call Bell label $\ell$ active in sector $q_A$ only when $w_\ell(q_A)>0$. This activity condition is used in Sec.~\ref{sec:screen}.

\subsection{Even--odd mode symmetry and modulo-four support separation}
\label{subsec:mod4}

The tensor-product form $U_{\mathrm h}^{(N)}=U_{\mathrm{BS}}^{\otimes N}\otimes\mathbb I_2$ shows that the final binary mode index is unchanged by the half analyzer. Equivalently, after grouping even- and odd-index physical modes, the half analyzer consists of two identical $U_{\mathrm{BS}}^{\otimes N}$ blocks and does not mix the two mode-index classes. The Bell layer also couples only equal-index-parity modes through $(b_0,b_2)$ and $(b_1,b_3)$. The complete analyzer therefore separately preserves the photon numbers in the even- and odd-index mode classes.

For an output occupation record $\mathbf n$, define the alternating residue
\begin{equation}
\nu(\mathbf n)
=
\left[
\sum_{j=0}^{M-1}(-1)^j n_j
\right]_4
\in\{0,1,2,3\},
\label{eq:alternating_residue}
\end{equation}
where $[x]_4$ denotes the representative of $x$ modulo four in $\{0,1,2,3\}$. The preceding even--odd mode separation implies that the analyzer conserves $\nu$. Although the two integer populations are conserved separately, they are not fixed across the nonzero Fock branches of the coherent ancillary resources used here; the modulo-four difference is the branch-independent quantity common to those branches.

Assume that every nonzero Fock branch of the ancillary input has the same alternating residue $\nu_{\mathrm{anc}}$; the Bell modes are unoccupied in this ancillary state. Since $M/2$ is even, $b_0,b_2$ are even-indexed and $b_1,b_3$ are odd-indexed. Each component of $\psi^\pm$ adds zero residue. Each component of $\phi^\pm$ adds $\pm2\equiv2\pmod4$. Conservation of $\nu$ gives the exact necessary support conditions
\begin{equation}
\begin{aligned}
c_{\psi^\pm}(\mathbf n)\neq0
&\;\Longrightarrow\;
\nu(\mathbf n)=\nu_{\mathrm{anc}},\\
c_{\phi^\pm}(\mathbf n)\neq0
&\;\Longrightarrow\;
\nu(\mathbf n)=[\nu_{\mathrm{anc}}+2]_4.
\end{aligned}
\label{eq:label-residue-classes}
\end{equation}
These conditions are not sufficient, since destructive interference can still produce zero amplitudes within an allowed residue class.

At fixed total photon number $Q$, the two residue classes in Eq.~\eqref{eq:label-residue-classes} correspond to opposite parity of the total population in the even-index modes, equivalently in the odd-index modes. This is the same even/odd photon-counting structure used in the Bell analyzer of Ewert and van Loock~\cite{EwertvanLoock2014}.

Combining Eq.~\eqref{eq:label-residue-classes} with the root photon-number sectors of Eq.~\eqref{eq:root_charge_sectors} removes overlap between the $\psi$ and $\phi$ labels for a definite ancillary residue. The central sector contains only $\psi^-$. In the outer sectors, the residue separates $\psi^+$ from both $\phi$ labels. Every nonzero detector record supported by $\psi^+$ or $\psi^-$ is therefore excluded for the other three Bell labels. By normalization,
\begin{equation}
\mu_{\psi^+}=\mu_{\psi^-}=1 .
\label{eq:psi_unique_exact}
\end{equation}
For the resources used here, $\gamma_1$ has residue $2$, whereas each $\gamma_{\tau\ge2}$ has residue $0$; the recursive, $\gamma_1$-product, asymmetric, and mixed-block halves therefore have definite ancillary residue $2$. These selection rules do not distinguish $\phi^+$ from $\phi^-$; their remaining support ambiguity is treated in Sec.~\ref{sec:screen}.

\section{Directional support criterion and target-unique probability}
\label{sec:screen}

For the resources considered here, the photon number at the root and modulo-four selection rules of Sec.~\ref{sec:ttn} leave only the $\phi^+$--$\phi^-$ support overlap. Resolving it requires information about detector records from both halves of the analyzer. To do that we retain so-called zero or projective boundary directions for the rivals and Gram matrices for the target state resolved according to the records. 

In the following subsections we describe how we construct the records, that the Gram matrices we consider are rank one which implies that it is enough to check vanishing amplitudes of the rival state in each photon number sector in halves of the analyzer separately. Finally we describe the accumulation of the unique probability masses for the target Bell state and discuss the completness of the procedure. 

\subsection{Active rivals and directional half-tree records}

Fix a root sector $q_A$, with $q_B=Q-q_A$, and choose an active target Bell label $t$, so $w_t(q_A)>0$. Let $u_t(q_A)$ denote the target-unique probability mass in this sector. We denote by $\mathcal K_t(q_A)$ the active rivals that remain after the exact root-photon-number and modulo-four selection rules. In the unresolved Bell sectors, the only retained rival of $\phi^+$ is $\phi^-$, and conversely. If $\mathcal K_t(q_A)=\varnothing$, no rival remains and the complete target-sector mass is unique,
\begin{equation}
 u_t(q_A)=w_t(q_A).
\label{eq:empty_rival_sector_mass}
\end{equation}

At this stage, the target and its rivals play different roles. For a rival, only whether the final amplitude vanishes is required. Because the contraction above a subtree is linear in its boundary vector, multiplying a nonzero rival boundary vector by a nonzero scalar cannot alter any later zero test. A binary zero/nonzero flag is nevertheless insufficient, because two noncollinear nonzero boundary vectors can couple differently to the remainder of the tree. For a detector record $\mathbf n_v$ below subtree $v$, define the rival record
\begin{equation}
\zeta_v^{(k)}(\mathbf n_v)
=
\begin{cases}
0,
&
\mathbf f_{k,v}^{(q_v)}(\mathbf n_v)=0,
\\[1mm]
\bigl[\mathbf f_{k,v}^{(q_v)}(\mathbf n_v)\bigr],
&
\mathbf f_{k,v}^{(q_v)}(\mathbf n_v)\neq0,
\end{cases}
\label{eq:rival-directional-class}
\end{equation}
where
\[
[\mathbf f]=\{\xi\mathbf f:\xi\in\mathbb C,\ \xi\neq0\}
\]
denotes the projective class of a nonzero vector.

For a retained-rival set $\mathcal K_t(q_A)$, we define the directional record of subtree $v$
\begin{equation}
\mathfrak d_v(\mathbf n_v)
=
\left(
 q_v,
 \bigl(\zeta_v^{(k)}(\mathbf n_v)\bigr)_{k\in\mathcal K_t(q_A)}
\right).
\label{eq:directional_record}
\end{equation}
The photon number $q_v$ is retained so that records from different photon-number sectors are never merged.

The rival part of a directional record is propagated through the same photon-number-resolved boundary recursion as in Sec.~\ref{subsec:ttn-boundary}, with each nonzero boundary vector retained only through its projective class.
For a directional record $\mathfrak d$ in photon-number sector $q$, define the target Gram matrix resolved by the record
\begin{equation}
\mathsf D_{t,v}(\mathfrak d)
=
\sum_{\substack{
\mathbf n_v\in\Omega_{v,q}\\
\mathfrak d_v(\mathbf n_v)=\mathfrak d
}}
\mathbf f_{t,v}^{(q)}(\mathbf n_v)
\mathbf f_{t,v}^{(q)}(\mathbf n_v)^\dagger.
\label{eq:record_target_gram}
\end{equation}
Summing over all directional records of a fixed photon number recovers the ordinary target environment,
\begin{equation}
\sum_{\mathfrak d:\,q_v=q}
\mathsf D_{t,v}(\mathfrak d)
=
\Eenv_{t,v}(q).
\label{eq:record_target_gram_sum}
\end{equation}
The matrices $\mathsf D_{t,v}$ propagate by the same double-layer contraction as the Gram environments of Sec.~\ref{subsec:root-sectors}. The rival transfer determines the parent directional record, and target Gram contributions yielding the same record are added.

\subsection{Rank-one half-root support and the rival-zero criterion}
\label{subsec:rankone-support}

Fix an active overlap sector and let $\ell\in\{t\}\cup\mathcal K_t(q_A)$ denote the target or a retained rival. For the factorized resources reported in Sec.~\ref{sec:results}, projection onto an outer root-photon-number sector selects one Bell-layer branch of Eq.~\eqref{eq:bell_layer_action}. The projected physical state factorizes across $A\mid B$, and the two half analyzers act independently. The zero criterion we define below is nevertheless imposed directly on the half-root Gram rank of each retained rival, not on the represented TTN sector dimension.

Let $q_X$, with $X\in\{A,B\}$,  be the corresponding half-root photon number. At the half roots, we also write $\zeta_X^{(k)}\equiv\zeta_{v_X}^{(k)}$. When the supported half-root space is one-dimensional,
\begin{equation}
\begin{aligned}
\Eenv_{\ell,X}(q_X)
&=
\eta_{\ell,X}(q_X)\,
\mathbf e_{\ell,X}(q_X)
\mathbf e_{\ell,X}(q_X)^\dagger,
\\
\eta_{\ell,X}(q_X)&>0,
\qquad
\|\mathbf e_{\ell,X}(q_X)\|=1,
\end{aligned}
\label{eq:rankone_environment_condition}
\end{equation}
with $\eta_{\ell,X}(q_X)=\tr\Eenv_{\ell,X}(q_X)$. Rank one here refers to the rank of the Gram matrix. It does not require the represented sector dimension $r_{\ell,X}^{(q_X)}$, or any internal TTN bond dimension, to equal one.

For any $\mathbf h\perp\mathbf e_{\ell,X}(q_X)$, the Gram matrix definition implies
\begin{equation}
0
=
\mathbf h^\dagger\Eenv_{\ell,X}(q_X)\mathbf h
=
\sum_{\mathbf n_X\in\Omega_{X,q_X}}
\left|
\mathbf h^\dagger
\mathbf f_{\ell,X}^{(q_X)}(\mathbf n_X)
\right|^2.
\label{eq:rankone_gram_argument}
\end{equation}
Every term of the sum is nonnegative, so every pattern-resolved boundary vector lies on the supported line,
\begin{equation}
\mathbf f_{\ell,X}^{(q_X)}(\mathbf n_X)
=
\xi_{\ell,X}(\mathbf n_X)\,
\mathbf e_{\ell,X}(q_X),
\qquad
\xi_{\ell,X}(\mathbf n_X)\in\mathbb C.
\label{eq:rankone_pattern_line}
\end{equation}
Define the scalar factor at the root
\begin{equation}
G_\ell(q_A)
=
\mathbf e_{\ell,A}(q_A)^T
\Rt_\ell(q_A)
\mathbf e_{\ell,B}(q_B).
\label{eq:rankone_root_factor}
\end{equation}
The superscript $T$ denotes a transposition. Substitution into the root amplitude and sector-mass formulas gives
\begin{equation}
\begin{aligned}
c_\ell(\mathbf n_A,\mathbf n_B)
&=
\xi_{\ell,A}(\mathbf n_A)
\xi_{\ell,B}(\mathbf n_B)
G_\ell(q_A),
\\
w_\ell(q_A)
&=
|G_\ell(q_A)|^2
\eta_{\ell,A}(q_A)
\eta_{\ell,B}(q_B).
\end{aligned}
\label{eq:rankone_amp}
\end{equation}

For the zero test below, the rank-one condition is required only for each retained rival. The target need not have rank-one half-root support for the exact probability assembly of Sec.~\ref{subsec:target-unique}; rank-one target support simplifies that assembly to a scalar form for the reported sectors. For a retained rival $k\in\mathcal K_t(q_A)$, being active means $w_k(q_A)>0$. Since $\eta_{k,A}(q_A)$ and $\eta_{k,B}(q_B)$ in Eq.~\eqref{eq:rankone_amp} are positive, $G_k(q_A)\neq0$ follows and is not an additional hypothesis and Eq.~\eqref{eq:rankone_pattern_line} gives $\mathbf f_{k,X}^{(q_X)}(\mathbf n_X)=0$ exactly when $\xi_{k,X}(\mathbf n_X)=0$, while Eq.~\eqref{eq:rival-directional-class} gives $\mathbf f_{k,X}^{(q_X)}(\mathbf n_X)=0$ exactly when $\zeta_X^{(k)}(\mathbf n_X)=0$. Hence,
\begin{equation}
c_k(\mathbf n_A,\mathbf n_B)=0
\quad\Longleftrightarrow\quad
\zeta_A^{(k)}(\mathbf n_A)
\ \text{or}\ 
\zeta_B^{(k)}(\mathbf n_B)=0.
\label{eq:half_zero_completeness}
\end{equation}
Thus, in an active rank-one rival sector, two nonzero rival half-boundary records cannot cancel at the global root.

The rank-one condition is essential for this guarantee. In a higher-dimensional supported space, nonzero half-boundary vectors can still have a vanishing root contraction. Such sectors require joint left--right information or an additional structural condition excluding root cancellations. The rank-one hypotheses above are exact mathematical conditions; the finite-precision rank diagnostics used in the reported calculations serve only as applicability checks.

\subsection{Target-unique mass and completeness}
\label{subsec:target-unique}

Consider first the single-rival case relevant to the unresolved Bell sectors, with retained rival $k$. At the half root $v_X$, write $\mathsf D_{t,X}\equiv\mathsf D_{t,v_X}$. The target records for which the rival vanishes on half $X\in\{A,B\}$ define the zero-filtered target Gram matrix
\begin{equation}
F_{t,X}^{(k,0)}(q)
=
\sum_{\substack{
\mathfrak d_X=(q,\ldots)\\
\zeta_X^{(k)}=0
}}
\mathsf D_{t,X}(\mathfrak d_X).
\label{eq:target-half-zero-gram}
\end{equation}
No rank condition on the target is required in this definition. By Eq.~\eqref{eq:half_zero_completeness}, the target record is unique against $k$ exactly when the rival vanishes on half $A$, on half $B$, or on both. Inclusion--exclusion over these two events therefore gives
\begin{align}
u_t(q_A)
={}&
\tr\!\left[
F_{t,A}^{(k,0)}(q_A)^T
\Rt_t(q_A)
\Eenv_{t,B}(q_B)
\Rt_t(q_A)^\dagger
\right]
\nonumber\\
&+
\tr\!\left[
\Eenv_{t,A}(q_A)^T
\Rt_t(q_A)
F_{t,B}^{(k,0)}(q_B)
\Rt_t(q_A)^\dagger
\right]
\nonumber\\
&-
\tr\!\left[
F_{t,A}^{(k,0)}(q_A)^T
\Rt_t(q_A)
F_{t,B}^{(k,0)}(q_B)
\Rt_t(q_A)^\dagger
\right].
\label{eq:one-rival-sector-mass-general}
\end{align}
Equation~\eqref{eq:one-rival-sector-mass-general} is exact for arbitrary target half-root rank. It uses the same target Gram matrices resolved according to the records already propagated by Eq.~\eqref{eq:record_target_gram}; the rank-one condition enters only through the rival-zero criterion.

For the reported overlap sectors, the target half-root support is also rank-one. Defining
\begin{equation}
z_{t,X}^{(k)}(q)=\tr F_{t,X}^{(k,0)}(q),
\label{eq:target-half-zero-mass}
\end{equation}
$F_{t,X}^{(k,0)}(q)$ is a positive semi-definite partial sum of $\Eenv_{t,X}(q)$, so Eq.~\eqref{eq:rankone_environment_condition} implies
\[
F_{t,X}^{(k,0)}(q)
=
z_{t,X}^{(k)}(q)\,
\mathbf e_{t,X}(q)\mathbf e_{t,X}(q)^\dagger.
\]
Using this relation together with Eq.~\eqref{eq:rankone_amp}, Eq.~\eqref{eq:one-rival-sector-mass-general} reduces to
\begin{equation}
\begin{aligned}
u_t(q_A)
=
|G_t(q_A)|^2
\Big[
&z_{t,A}^{(k)}(q_A)\eta_{t,B}(q_B)
+
\eta_{t,A}(q_A)z_{t,B}^{(k)}(q_B)
\\
&-
z_{t,A}^{(k)}(q_A)z_{t,B}^{(k)}(q_B)
\Big].
\end{aligned}
\label{eq:one-rival-sector-mass}
\end{equation}
This scalar form follows from rank-one target factorization and is the characteristic used for the reported overlap sectors.

For a finite retained-rival set, a target record is unique exactly when every $k\in\mathcal K_t(q_A)$ satisfies the half-zero criterion of Eq.~\eqref{eq:half_zero_completeness}. The following proposition states when the directional construction gives this mass exactly.

\begin{proposition}
\label{prop:directional_exactness}
\emph{[Directional completeness under rank-one rivals.]} Fix an active target $t$ in root sector $q_A$, and let $\mathcal K_t(q_A)$ contain every active rival surviving the exact preliminary selection rules. Assume an exact TTN representation and exact zero and projective-equivalence decisions. If $\mathcal K_t(q_A)=\varnothing$, then $u_t(q_A)=w_t(q_A)$. If $\mathcal K_t(q_A)\neq\varnothing$ and every retained rival has a one-dimensional supported space at both half roots, the directional construction with the record-resolved target Gram matrices gives the exact target-unique sector mass $u_t(q_A)$.
\end{proposition}
The proof is given in Appendix~\ref{app:general-records}.

Thus $u_t(q_A)=0$ for an inactive target, $u_t(q_A)=w_t(q_A)$ when no rival remains, and otherwise the directional construction gives the sector contribution under Proposition~\ref{prop:directional_exactness}; Eq.~\eqref{eq:one-rival-sector-mass-general} gives the exact single-rival assembly, with Eq.~\eqref{eq:one-rival-sector-mass} its rank-one-target specialization for the reported overlap sectors. Since the root-photon-number sectors are mutually disjoint,
\begin{equation}
\mu_t
=
\sum_{q_A=0}^{Q}u_t(q_A).
\label{eq:final_screen_sum}
\end{equation}
This is the Bell-label unique mass entering the equal-prior discrimination probability defined in Sec.~\ref{sec:model}.

\section{Results and validation}
\label{sec:results}

We evaluate the support construction of Secs.~\ref{sec:ttn} and~\ref{sec:screen} for the recursive and $\gamma_1$-product resources described in Sec.~\ref{sec:model}, together with the asymmetric recursive--product and mixed-block compositions defined below. Throughout this section, a hat denotes a finite-precision TTN value obtained with the numerical decision rules specified in Appendix~\ref{app:numerics}. Exact fractions are quoted only for analytical results established here or for exact cited benchmarks whose applicability to the stated analyzer--resource pair is established explicitly.

\subsection{Bell-label and root-sector-resolved results}
\label{subsec:resource-results}

Table~\ref{tab:success-probabilities} summarizes the Bell-label unique masses and aggregate discrimination probabilities. For the definite-residue resources in the table, Sec.~\ref{subsec:mod4} gives $\mu_{\psi^+}=\mu_{\psi^-}=1$, so only the two $\phi$-label masses are shown. For the asymmetric recursive--product instances, half $A$ carries the recursive resource and half $B$ the $\gamma_1$-product resource.

\begin{table*}[t]
\caption{Bell-label unique masses and success probabilities for the analyzer--resource pairs considered here. Hatted quantities are finite-threshold numerical TTN evaluations, and the numerical columns are rounded independently to six decimal places from the unrounded TTN values. Since $\mu_{\psi^+}=\mu_{\psi^-}=1$, only the two $\phi$-label masses are shown. The final column gives an analytical or literature reference value for $P_{\mathrm{USD}}$.}
\label{tab:success-probabilities}
\begin{ruledtabular}
\begin{tabular}{lccccc}
ancillary resource
& $M$
& $\widehat{\mu}_{\phi^+}$
& $\widehat{\mu}_{\phi^-}$
& $\widehat P_{\mathrm{USD}}$
& analytical/reference $P_{\mathrm{USD}}$
\\
\hline
recursive $=$ $\gamma_1$ product
& $8$
& $0.500000$
& $0.500000$
& $0.750000$
& $\frac34$\textsuperscript{a,b}
\\
recursive
& $16$
& $0.750000$
& $0.750000$
& $0.875000$
& $\frac78$\textsuperscript{a}
\\
$\gamma_1$ product
& $16$
& $0.375000$
& $0.750000$
& $0.781250$
& $\frac{25}{32}$\textsuperscript{b}
\\
asymmetric recursive--product
& $16$
& $0.562500$
& $0.750000$
& $0.828125$
& $\frac{53}{64}$\textsuperscript{e}
\\
recursive
& $32$
& $0.875000$
& $0.875000$
& $0.937500$
& $\frac{15}{16}$\textsuperscript{a}
\\
$\gamma_1$ product
& $32$
& $0.273438$
& $0.875000$
& $0.787109$
& $\frac{403}{512}$\textsuperscript{c}
\\
asymmetric recursive--product
& $32$
& $0.574219$
& $0.875000$
& $0.862305$
& $\frac{883}{1024}$\textsuperscript{e}
\\
reflection-covariant mixed-block
& $32$
& $0.687500$
& $0.875000$
& $0.890625$
& $\frac{57}{64}$\textsuperscript{d}
\end{tabular}
\end{ruledtabular}

\begin{minipage}{\textwidth}
\vspace{1mm}
\footnotesize
\textsuperscript{a}
The recursive-family analytical values are from Ewert and van Loock~\cite{EwertvanLoock2014}.
\quad
\textsuperscript{b}
The $\gamma_1$-product aggregate values $3/4$ and $25/32$ are from Ewert and van Loock~\cite{EwertvanLoock2014}.
\quad
\textsuperscript{c}
The value $403/512$ follows from Corollary~9 of the cited arXiv version of Yamazaki \textit{et al.}~\cite{YamazakiIkutaYamamoto2023}. Its applicability to the present endpoint is established analytically in this work.
\quad
\textsuperscript{d}
Supplemental Lemma~A3 of Ewert and van Loock~\cite{EwertvanLoock2014}, with its index set to $2$, gives $57/64$ for $\gamma_1\gamma_2^{\otimes3}$; reflection gives the ordering used here. The displayed $\phi$ masses are numerical TTN values.
\quad
\textsuperscript{e}
The exact aggregate follows from the half-resource composition identity derived in this work.
\end{minipage}
\end{table*}

The recursive rows reproduce their analytical benchmarks. The $M=16$ $\gamma_1$-product row reproduces the Ewert--van Loock aggregate and sign-resolved results~\cite{EwertvanLoock2014}. For $M=32$, the product endpoint is analytically equivalent, for photon-counting support, to the eight-position construction of Yamazaki \textit{et al.}~\cite{YamazakiIkutaYamamoto2023}. Its exact aggregate is therefore $403/512$, agreeing with $\widehat P_{\mathrm{USD}}\approx0.787109$. The hatted sign-resolved masses remain numerical TTN outputs. The endpoint mapping and the half-resource composition proof are given in Appendix~\ref{app:half-additivity}.

The TTN calculations for the asymmetric recursive--product cases do not use the half-resource composition identity. Independently, the reflected-endpoint identity gives the exact aggregates $53/64$ for $M=16$ and $883/1024$ for $M=32$, in agreement with Table~\ref{tab:success-probabilities}. The sign-resolved masses and root-sector contributions remain hatted numerical TTN-based quantities.
\begin{table*}[t]
\caption{TTN and retained-directional-record diagnostics for the reference numerical profile. The quantity $r_{\mathrm{sec}}^{\mathrm{peak}}$ is a represented photon-number-sector dimension, not the rank of the supported half-root Gram space used in Sec.~\ref{sec:screen}. These finite-instance diagnostics depend on the TTN representation, update schedule, and finite-precision profile and are not physical or asymptotic invariants.}
\label{tab:ttn-ranks}
\centering
\small
\begin{ruledtabular}
\begin{tabular}{lcccc}
ancillary resource
& $M$
& $\chi_{\mathrm{peak}}$
& $r_{\mathrm{sec}}^{\mathrm{peak}}$
& $S_{\max}$
\\
\hline
recursive $=$ $\gamma_1$ product & 8 & 4 & 2 & 4\\
recursive & 16 & 12 & 3 & 13\\
$\gamma_1$ product & 16 & 16 & 6 & 22\\
asymmetric recursive--product & 16 & 16 & 6 & 20\\
recursive & 32 & 207 & 31 & 660\\
$\gamma_1$ product & 32 & 512 & 73 & 5462\\
asymmetric recursive--product & 32 & 512 & 72 & 5457\\
reflection-covariant mixed-block & 32 & 322 & 44 & 1146\\
\end{tabular}
\end{ruledtabular}
\end{table*}
For the reflection-covariant $M=32$ mixed-block resource, half $A$ contains $\ket{\gamma_2}^{\otimes3}\otimes\ket{\gamma_1}$ and half $B$ the reflected ordering $\ket{\gamma_1}\otimes\ket{\gamma_2}^{\otimes3}$, with $\gamma_1$ adjacent to the Bell modes on both sides. Here ``mixed-block'' refers only to the different ancillary block sizes; the state itself is pure. The TTN calculation gives $\widehat P_{\mathrm{USD}}\approx0.890625$, numerically agreeing with the exact value $57/64$ obtained from Supplemental Lemma~A3 of Ewert and van Loock~\cite{EwertvanLoock2014}, with $\widehat\mu_{\phi^+}\approx0.687500$ and $\widehat\mu_{\phi^-}\approx0.875000$. No exact rational values are inferred for these individual hatted masses from their decimal representations.

At $M=32$, the populated root-photon-number sectors are $q_A=14,15,16$. The central sector $q_A=15$ isolates $\psi^-$. For the asymmetric orientation defined above, $q_A=14$ carries the product-half outer contribution, whereas $q_A=16$ carries the recursive-half outer contribution. The nontrivial numerical TTN contributions in the two outer sectors are
\begin{equation}
\begin{aligned}
\widehat u_{\phi^+}(14)&\approx0.136719,
&\widehat u_{\phi^-}(14)&\approx0.437500,\\
\widehat u_{\phi^+}(16)&\approx0.437500,
&\widehat u_{\phi^-}(16)&\approx0.437500.
\end{aligned}
\label{eq:asym-root-sector-phi}
\end{equation}
Their sums give $\widehat\mu_{\phi^+}\approx0.574219$ and $\widehat\mu_{\phi^-}\approx0.875000$, reproducing the Bell-label masses in Table~\ref{tab:success-probabilities}. The $\psi^-$ central contribution and the unit total masses of both $\psi$ labels follow from the exact photon-number and residue selection rules of Secs.~\ref{subsec:root-sectors} and~\ref{subsec:mod4}.

The arithmetic-mean identity of Appendix~\ref{app:half-additivity} relates the aggregate asymmetric success probability to the two symmetric endpoint probabilities. The more detailed half-resource decomposition in the same appendix identifies each outer-sector contribution through the corresponding half-resource support quantity. The TTN calculation evaluates these contributions directly, including the unequal $\phi^+$ contributions of the product-half and recursive-half outer sectors in Eq.~\eqref{eq:asym-root-sector-phi}, and the independent half-circuit calculation provides a separate check. For the reflection-covariant mixed-block resource, the Bell-label-resolved contributions in the two reflection-related outer sectors agree within numerical precision.

Beyond the aggregate and root-sector-resolved quantities above, the evolved Bell-labeled TTNs can also be contracted directly against a specified physical PNRD record without constructing the complete detector-support table.

\subsection{Numerical validation}
\label{subsec:numerical-checks}

Table~\ref{tab:ttn-ranks} reports our finite-instance diagnostics for the TTN representation and directional transfer. Here $\chi_{\mathrm{peak}}$ is the largest total TTN bond dimension encountered, $r_{\mathrm{sec}}^{\mathrm{peak}}$ the largest represented photon-number-sector dimension, and $S_{\max}$ the largest retained directional-record count at any tracked half-tree node.

For the reported one-rival overlap sectors, at most two directional records are retained at a half root for a fixed charge. These are the zero class and, when present, one nonzero projective class. This does not mean that only two detector records exist or that the half tree carries only a binary directional state; larger $S_{\max}$ values indicate that many distinct directional records can occur in intermediate subtrees. Likewise, rank-one completeness concerns the retained-rival half-root Gram spaces and does not require $r_{\mathrm{sec}}^{\mathrm{peak}}=1$. These are finite-instance representation diagnostics; no asymptotic bound on the retained-record count or represented sector dimensions is established here.

At $M=8$ the recursive and product resources coincide, and the $Q=6$ fixed-photon-number detector space contains $1716$ records. Complete finite-threshold patternwise classification reproduces the Bell-label masses obtained from the directional construction. With the modulo-four selection rule disabled, some outer target sectors retain two rivals; the finite-rival directional assembly again agrees with an independent permanent evaluation of all $1716$ detector records.

At $M=32$, five $\gamma$-block compositions were compared between the TTN support calculation and a separately implemented half-circuit calculation.

For the reported overlap calculations, the numerical applicability checks include resolved support and projective decisions, the required half-root ranks and retained-rival root factors, threshold separation, and SVD discarded-weight bounds. State normalization and probability-range conditions are checked separately. Analytical benchmark fractions are used only for post-calculation comparison and are not inputs to the TTN support calculation. Passing these checks supports the reported hatted quantities at the stated finite-precision profile; it does not turn numerical zeros into exact support statements. Exact values are asserted only when established analytically or supplied by an exact cited result. Detailed permanent checks, an exhaustive $M=8$ comparison, a threshold sweep, an update-path comparison, and an independent $M=32$ comparison are presented in Appendix~\ref{app:validation}.

\section{Discussion}
\label{sec:discussion}

For the definite-residue resources considered here, the symmetry analysis reduces the detector-support problem to a small set of active Bell-label sectors. Root photon number isolates the central $\psi^-$ sector, while the modulo-four residue separates $\psi^+$ from $\phi^\pm$ in the outer sectors. The only remaining Bell-label competition is therefore between $\phi^+$ and $\phi^-$ in those outer sectors. The exact completeness result assumes one-dimensional retained-rival support; in the reported calculations, this condition is assessed numerically at the stated finite-precision profile and serves as the applicability check for the hatted results. The TTN provides a photon-number-resolved representation in which the required boundary data can be propagated. Still, the completeness result follows from this support geometry rather than from the tensor-network representation itself.

The half-resource composition relation gives a complementary analytical statement. When the ancillary input factorizes across the two analyzer halves, and each half has definite photon number and definite alternating residue, the two outer branches probe the two halves separately. Hence, the $\phi$ unique masses combine additively. The half-resource decomposition identifies the two outer-sector contributions through the corresponding support quantities. At the same time, the reflected-endpoint corollary gives the arithmetic mean of the two symmetric endpoint probabilities. It shows that the asymmetric pairing cannot exceed the better endpoint for the same analyzer. In the reported calculations, these support quantities are evaluated with the TTN and independently checked with the half-circuit calculation.

The finite-instance calculations do not establish an asymptotic computational advantage, because no scaling bounds are proved for the retained directional-record counts or represented photon-number-sector dimensions. The reported success probabilities are evaluations for the specified passive analyzers and ancillary resources under lossless optics, ideal ancillary preparation, indistinguishable photons, ideal PNRDs, and equal Bell priors. They are not global optima over passive interferometers or ancillary resources, and they do not provide bounds on more general Bell measurements.

If a retained rival has support dimension greater than one at either half root, a bilinear root contraction can vanish even when both half-boundary vectors are nonzero. Hence, the separated-half zero criterion is no longer complete. A higher-rank treatment therefore requires additional joint information from the two halves, unless another structural condition excludes such cancellations. Photon loss~\cite{Wein2016}, partial distinguishability~\cite{Shchesnovich2015}, detector inefficiency, and dark counts~\cite{Wein2016} change the probabilities of the detector records considered here. Records that have exactly zero probability in the ideal model can then acquire nonzero probability, so the present support-defined USD criterion is not directly applicable without modification. An error-tolerant discrimination objective would be a different problem.

\section*{Data availability}
The numerical data and python code supporting the findings of this study are available from the corresponding author upon reasonable request.

\begin{acknowledgments}
AK acknowledges financial support from the Institution of Eminence (IoE), Banaras Hindu University, Varanasi, under the ``International Visiting Student Program'' scheme, and from the Council of Scientific and Industrial Research (CSIR), Government of India, under Award No.~09/0013(21020)/2025-EMR-I. DKM acknowledges the Chanakya Doctoral Fellowship (Grant No.~I-HUB/DF/2022-23/04). WR and MT acknowledge support from JST Moonshot R\&D, Grant No. JPMJMS226C and Grant No. JPMJMS2061, JST ASPIRE, Grant No. JPMJAP2427, JST COINEXT Grant No. JPMJPF2221, and  JST CRONOS, Grant No. JPMJCS24N6.
\end{acknowledgments}

\appendix

\section{Directional contraction and supporting calculations}

\subsection{Photon-number-resolved TTN updates and environment recursions}\label{app:ttn_updates}

This subsection presents the photon-number-resolved tensor updates used in the optical evolution and the recursive double-layer contraction that underlies Eq.~\eqref{eq:general_environment_u1}. Let $\mathsf B$ denote the balanced-beam-splitter tensor in the local Fock basis of Sec.~\ref{sec:ttn},
\[
\mathsf B^{m_1m_2}_{n_1n_2}
=
\bra{n_1,n_2}\Gamma(U_{\mathrm{BS}})\ket{m_1,m_2}.
\]
Photon-number conservation gives $n_1+n_2=m_1+m_2$. When the acted-on modes are sibling leaves of a lowest-level node $v$, the gate can therefore be applied separately in each fixed pair-photon-number block.
\[
\bigl[\widetilde{\mathsf T}_v^{(\ell)}\bigr]^{(q,\omega)}_{n_1n_2}
=
\sum_{\substack{0\le m_1,m_2<d_{\mathrm{loc}}\\m_1+m_2=q}}
\mathsf B^{m_1m_2}_{n_1n_2}
\bigl[\mathsf T_v^{(\ell)}\bigr]^{(q,\omega)}_{m_1m_2}.
\]
Here $0\le n_1,n_2<d_{\mathrm{loc}}$ and $n_1+n_2=q$; matrix elements outside that output sector vanish by photon-number conservation. For a nonsibling gate, the affected path is reshaped across the chosen cut. Let $\mathsf M$ denote the resulting matrix. If $q$ denotes the photon number carried by the refactorized bond, photon-number grading gives
\[
\mathsf M
=
\bigoplus_{q\ge0}\mathsf M^{(q)},
\qquad
\mathsf M^{(q)}
=
\mathsf U^{(q)}\boldsymbol\Sigma^{(q)}\mathsf W^{(q)\dagger}.
\]
In exact arithmetic, retaining every nonzero singular value ensures that this refactorization is exact. In the stored TTN, the singular values are absorbed into a single factor, and the other factor is kept isometric. The orientation follows the local routed split, as in the one-sided SVD gauge described in Sec.~\ref{sec:ttn}; no global root-canonical gauge is assumed.

The Gram environment at a lowest-level node joining physical modes $i$ and $j$ is
\[
\bigl[\Eenv_{\ell,v}(q_v)\bigr]^{\omega}_{\omega'}
=
\sum_{\substack{0\le n_i,n_j<d_{\mathrm{loc}}\\n_i+n_j=q_v}}
[\mathsf T_v^{(\ell)}]^{(q_v,\omega)}_{n_i n_j}
\left(
[\mathsf T_v^{(\ell)}]^{(q_v,\omega')}_{n_i n_j}
\right)^* .
\]
For a higher non-root node, primed indices denote independent bra-layer copies, and all compatible child photon-number partitions are summed.
\begin{align}
\bigl[\Eenv_{\ell,v}(q_v)\bigr]^{\omega}_{\omega'}
={}&
\sum_{\substack{q_{\mathrm L}+q_{\mathrm R}=q_v\\
\alpha,\alpha',\beta,\beta'}}
[\mathsf T_v^{(\ell)}]^{(q_v,\omega)}
_{(q_{\mathrm L},\alpha)(q_{\mathrm R},\beta)}
\nonumber\\
&\times
\bigl[\Eenv_{\ell,v_{\mathrm L}}(q_{\mathrm L})\bigr]^{\alpha}_{\alpha'}
\bigl[\Eenv_{\ell,v_{\mathrm R}}(q_{\mathrm R})\bigr]^{\beta}_{\beta'}
\nonumber\\
&\times
\left(
[\mathsf T_v^{(\ell)}]^{(q_v,\omega')}
_{(q_{\mathrm L},\alpha')(q_{\mathrm R},\beta')}
\right)^* .
\label{eq:internal_environment_u1}
\end{align}
The child-photon-number sum is required because Eq.~\eqref{eq:general_environment_u1} sums over every subtree record of parent photon number $q_v$; unlike the record-specific recursion in Eq.~\eqref{eq:internal_boundary_amplitude}, a fixed parent photon number does not select a unique child partition.

\subsection{Finite-rival transfer and completeness proof}
\label{app:general-records}

At a lowest internal node $v$ joining the two physical modes $i$ and $j$, the target Gram matrix associated with one directional record is evaluated directly from the local tensor,
\begin{align}
\bigl[\mathsf D_{t,v}(\mathfrak d)\bigr]_{\omega\omega'}
={}&
\sum_{\substack{0\le n_i,n_j<d_{\mathrm{loc}}\\
n_i+n_j=q_v\\
\mathfrak d_v(n_i,n_j)=\mathfrak d}}
\bigl[\mathsf T_v^{(t)}\bigr]^{(q_v,\omega)}_{n_i n_j}
\left(
\bigl[\mathsf T_v^{(t)}\bigr]^{(q_v,\omega')}_{n_i n_j}
\right)^* .
\label{eq:app-record-leaf}
\end{align}
For a compatible pair of child records at a higher node,
\begin{align}
&\bigl[\mathsf Y_v^{(t)}(\mathfrak d_{v_{\mathrm L}},\mathfrak d_{v_{\mathrm R}})\bigr]_{\omega\omega'}
\nonumber\\[-1mm]
&=\sum_{\alpha,\alpha',\beta,\beta'}
\bigl[\mathsf T_v^{(t)}\bigr]^{(q_v,\omega)}_{(q_{\mathrm L},\alpha)(q_{\mathrm R},\beta)}
\bigl[\mathsf D_{t,v_{\mathrm L}}(\mathfrak d_{v_{\mathrm L}})\bigr]_{\alpha\alpha'}
\nonumber\\[-1mm]
&\quad\times
\bigl[\mathsf D_{t,v_{\mathrm R}}(\mathfrak d_{v_{\mathrm R}})\bigr]_{\beta\beta'}
\left(
\bigl[\mathsf T_v^{(t)}\bigr]^{(q_v,\omega')}_{(q_{\mathrm L},\alpha')(q_{\mathrm R},\beta')}
\right)^* ,
\label{eq:app-record-transfer}
\end{align}
where $q_{\mathrm L}+q_{\mathrm R}=q_v$.

For each retained rival $k\in\mathcal K_t(q_A)$, choose vectors $\mathbf g_{k,\mathrm L}$ and $\mathbf g_{k,\mathrm R}$ representing the two child records. Use the zero vector for a zero record and any representative for a nonzero projective class. Propagate these vectors as
\begin{equation}
\bigl[\mathbf g_{k,v}\bigr]_{\omega}
=
\sum_{\alpha,\beta}
\bigl[\mathsf T_v^{(k)}\bigr]^{(q_v,\omega)}_{(q_{\mathrm L},\alpha)(q_{\mathrm R},\beta)}
\bigl[\mathbf g_{k,\mathrm L}\bigr]_{\alpha}
\bigl[\mathbf g_{k,\mathrm R}\bigr]_{\beta}.
\label{eq:app-rival-record-transfer}
\end{equation}
The parent rival record is zero when $\mathbf g_{k,v}=0$ and is $[\mathbf g_{k,v}]$ otherwise. A nonzero rescaling of either child representative changes $\mathbf g_{k,v}$ only by an overall nonzero factor, so the parent zero/nonzero status and projective class are independent of the chosen representatives. Child-record pairs that produce the same parent record are merged by summing their target Gram contributions,
\begin{equation}
\bigl[\mathsf D_{t,v}(\mathfrak d_v)\bigr]_{\omega\omega'}
=
\sum_{(\mathfrak d_{v_{\mathrm L}},\mathfrak d_{v_{\mathrm R}})\mapsto\mathfrak d_v}
\bigl[\mathsf Y_v^{(t)}(\mathfrak d_{v_{\mathrm L}},\mathfrak d_{v_{\mathrm R}})\bigr]_{\omega\omega'}.
\label{eq:app-record-merge}
\end{equation}
Summing over all directional records at fixed photon number removes the record constraint and recovers the ordinary target environment.

If no rival remains, the result follows directly. We therefore consider $\mathcal K_t(q_A)\neq\varnothing$ and assume that every retained rival satisfies the rank-one half-root condition of Proposition~\ref{prop:directional_exactness}. At a half root $v_X$, write $\mathsf D_{t,X}\equiv\mathsf D_{t,v_X}$ and denote the corresponding directional record by $\mathfrak d_X$. For a half $X\in\{A,B\}$ and a subset $J\subseteq\mathcal K_t(q_A)$, define the zero-filtered target Gram matrix
\begin{equation}
F_{t,X}(J;q)
=
\sum_{\substack{\mathfrak d_X=(q,\ldots)\\
\zeta_X^{(k)}=0\;\forall k\in J}}
\mathsf D_{t,X}(\mathfrak d_X).
\label{eq:finite-rival-filtered-gram}
\end{equation}
Thus $F_{t,X}(\varnothing;q)=\Eenv_{t,X}(q)$, while $F_{t,X}(\{k\};q)=F_{t,X}^{(k,0)}(q)$ of Eq.~\eqref{eq:target-half-zero-gram}. Inclusion--exclusion within one half gives the target Gram matrix carried by records for which every rival in $J$ has a nonzero half-root record,
\begin{equation}
F_{t,X}^{\mathrm{nz}}(J;q)
=
\sum_{I\subseteq J}(-1)^{|I|}F_{t,X}(I;q).
\label{eq:finite-rival-nz-gram}
\end{equation}
For $J=\{k\}$ this reduces to
$F_{t,X}^{\mathrm{nz}}=\Eenv_{t,X}-F_{t,X}^{(k,0)}$.

Under rank-one rival support, Eq.~\eqref{eq:half_zero_completeness} implies that a rival $k$ has a nonzero global amplitude on a detector record exactly when its half-root records are nonzero on both halves. For a fixed subset $J$ of rivals, the target probability carried by records on which every $k\in J$ is globally nonzero is therefore
\begin{equation}
\tr\!\left[
F_{t,A}^{\mathrm{nz}}(J;q_A)^T
\Rt_t(q_A)
F_{t,B}^{\mathrm{nz}}(J;q_B)
\Rt_t(q_A)^\dagger
\right].
\label{eq:finite-rival-nz-root-weight}
\end{equation}
Inclusion--exclusion over these globally nonzero rival events gives the exact target-unique sector mass
\begin{align}
u_t(q_A)
={}&
\sum_{J\subseteq\mathcal K_t(q_A)}(-1)^{|J|}
\tr\!\left[
F_{t,A}^{\mathrm{nz}}(J;q_A)^T
\Rt_t(q_A)
\right.
\nonumber\\[-1mm]
&\left.\hspace{33mm}\times
F_{t,B}^{\mathrm{nz}}(J;q_B)
\Rt_t(q_A)^\dagger
\right].
\label{eq:finite_rival_sector_mass}
\end{align}
No rank condition on the target is used in Eqs.~\eqref{eq:finite-rival-filtered-gram}--\eqref{eq:finite_rival_sector_mass}. If the target half-root support is rank one, each filtered Gram matrix is proportional to the same supported projector on that half; taking traces then reduces Eq.~\eqref{eq:finite_rival_sector_mass} to the scalar inclusion--exclusion form used for the reported sectors. For one retained rival, the matrix formula reduces to Eq.~\eqref{eq:one-rival-sector-mass-general}, and its rank-one-target specialization is Eq.~\eqref{eq:one-rival-sector-mass}.

Projective reduction and record merging preserve all information required by the support test. Once the occupations outside a subtree $v$ are fixed, contraction of the tensors above $v$ is a linear functional $\mathcal F_{k,v}^{(q_v)}$ of the rival boundary vector, so $\mathcal F_{k,v}^{(q_v)}(\xi\mathbf f)=\xi\mathcal F_{k,v}^{(q_v)}(\mathbf f)$ for every $\xi\neq0$. Thus, replacing a nonzero rival boundary vector by its projective class preserves every later zero test and Eq.~\eqref{eq:app-record-leaf} is the lowest-node Gram sum over exactly the patterns assigned to a directional record. At a higher node, every parent pattern has a unique ordered pair of child patterns, and Eqs.~\eqref{eq:app-record-transfer}--\eqref{eq:app-record-merge} sum all child pairs mapped to the same parent record. Induction over the tree, therefore, shows that the directional records partition the represented detector patterns without omission or double counting; restricting a transfer to photon-number blocks compatible with the requested parent photon number removes no patterns from that sector.

At the two half roots, Eq.~\eqref{eq:half_zero_completeness} implies that a retained rival has zero global amplitude exactly when its left or right record is zero. The Eqs.~\eqref{eq:finite-rival-nz-gram} and~\eqref{eq:finite_rival_sector_mass} therefore select exactly the target patterns for which every retained rival has zero global amplitude. A Bell label omitted from $\mathcal K_t(q_A)$ is either excluded by the exact preliminary selection rules or has zero total weight in the root sector. In the latter case the sector weight is a sum of nonnegative squared amplitudes, so every amplitude of that label vanishes in the sector. Thus the retained-rival conditions are sufficient to prove Proposition~\ref{prop:directional_exactness}. The proof assumes an exact TTN representation and exact zero/projective decisions.

\subsection{Finite-precision rules}
\label{app:numerics}

Hatted quantities denote numerical TTN evaluations; support-derived hatted masses additionally depend on the finite-threshold decisions specified here. The reference calculations use
\begin{equation}
(\varepsilon_{\mathrm{svd}},\varepsilon_{\mathrm{supp}},\varepsilon_0,\varepsilon_{\mathrm{merge}})
=
(10^{-13},10^{-12},10^{-10},10^{-7}),
\label{eq:numerical-reference-profile}
\end{equation}
with binary64 arithmetic, $\varepsilon_{\mathrm{mach}}=2^{-52}\simeq2.220446\times10^{-16}$. Here $\varepsilon_{\mathrm{svd}}$ controls photon-number-block singular-value retention, $\varepsilon_{\mathrm{supp}}$ the Gram-support and retained-rival root-factor decisions, $\varepsilon_0$ boundary-vector zero decisions, and $\varepsilon_{\mathrm{merge}}$ projective-class merging.

For a charge block $q$, let $s_{q,1}\ge s_{q,2}\ge\cdots$ be its singular values, $s_{\max}=\max_q s_{q,1}$, and let $D_{\mathrm{row}}\times D_{\mathrm{col}}$ be the full conceptual matrix shape before the photon-number decomposition. With
\[
\delta_{\mathrm{svd}}
=
\varepsilon_{\mathrm{mach}}\max(D_{\mathrm{row}},D_{\mathrm{col}})s_{\max},
\]
a singular value is retained when
\begin{equation}
s_{q,j}>
\max\!\left(
\varepsilon_{\mathrm{svd}}s_{q,1},
\delta_{\mathrm{svd}}
\right).
\label{eq:app-svd-rule}
\end{equation}
The relative term is local to block $q$, while the roundoff floor is common to the full factorization. The discarded absolute and relative squared-singular-value weights are recorded for each refactorization and are required to remain below $10^{-20}$.

Support eigenvalues, nonnegative support masses, boundary-vector norms, and retained-rival root factors are classified using three protected outcomes. A quantity is safely zero, unresolved, or safely nonzero. Roundoff floors proportional to machine precision, the relevant numerical scale, and the square root of the working dimension are included in these decisions. The requested support tolerance $\varepsilon_{\mathrm{supp}}$ is absolute rather than rescaled by the largest Gram eigenvalue; the retained-rival root-factor test instead scales this tolerance by the spectral norm of the corresponding root block. Before support classification, each numerical Gram matrix is checked for Hermiticity and positive semidefiniteness within the binary64 roundoff scale. All support and boundary-vector thresholds are evaluated in the chosen one-sided SVD gauge. If a decision required by the directional construction or scalar root assembly remains unresolved, the corresponding numerical result is not accepted.

A resolved nonzero boundary vector is normalized to unit norm and its global phase is fixed using its largest-magnitude component. Projective representatives are compared by Euclidean distance after this canonicalization. Two projective classes are identified only when their distance is safely below $\varepsilon_{\mathrm{merge}}$; a distance in the protected gap is treated as unresolved, and the corresponding numerical result is not accepted. The precise roundoff safety factors and projective-class comparison rules are specified in the numerical implementation used for the reported calculations.

\subsection{Numerical validation}
\label{app:validation}

Using the matrix convention of Sec.~\ref{sec:model}, let $U_{\mathrm{an}}[\mathbf n|\mathbf m]$ be the $Q\times Q$ matrix obtained by repeating input column $i$ exactly $m_i$ times and output row $j$ exactly $n_j$ times. For $\mathbf m,\mathbf n\in\Omega_{M,Q}$, the bosonic transition amplitude is~\cite{Scheel2008}
\begin{equation}
\langle\mathbf n|\Gamma(U_{\mathrm{an}})|\mathbf m\rangle
=
\frac{\operatorname{Per}\!\left(U_{\mathrm{an}}[\mathbf n|\mathbf m]\right)}
{\sqrt{\displaystyle\prod_{i=0}^{M-1}m_i!\prod_{j=0}^{M-1}n_j!}}.
\label{eq:permanent-validation-amplitude}
\end{equation}
Here $\operatorname{Per}$ denotes the matrix permanent. Bell-input amplitudes are obtained by linear superposition over their nonzero Fock branches. The permanent calculations are used only for independent validation.

The independent patternwise classifier uses $\varepsilon_{\mathrm{dec}}=10^{-10}$ with a factor-ten guard band. Values with $|\widehat c|\le10^{-11}$ are resolved zero, values with $|\widehat c|\ge10^{-9}$ are resolved nonzero, and intermediate values are unresolved. A detector record is certified as unique only when one Bell amplitude is resolved nonzero, and the other three are resolved zero; records involving unresolved amplitudes are not certified.

For the asymmetric $M=16$ resource, consider the physical detector record
\begin{equation}
\mathbf n_*
=
(0,1,1,2,1,2,0,1,\,
1,2,0,0,1,0,0,2),
\label{eq:specified-record-example}
\end{equation}
with $Q=14$ and $q_A=8$. The TTN contraction gives
\begin{equation}
\begin{aligned}
\left|\widehat c_{\phi^+}(\mathbf n_*)\right|
&\approx3.662109375\times10^{-4},\\
\max_{\ell\neq\phi^+}
\left|\widehat c_\ell(\mathbf n_*)\right|
&\approx5.76\times10^{-20}.
\end{aligned}
\label{eq:specified-record-certificate}
\end{equation}
The record is therefore classified as $\phi^+$. The selected $M=16$ validation set also contains one unique record for each of the other three Bell labels, one $\phi^+$--$\phi^-$ inconclusive record, and one numerically unsupported record. Across all six certificates, the maximum absolute TTN--permanent amplitude discrepancy is below $2.1\times10^{-16}$. These spot checks validate only the selected amplitudes; they do not, by themselves, validate the complete detector-support assembly.

At $M=8$ and $Q=6$, all $1716$ detector records were evaluated independently. Complete enumeration reproduces the directional Bell-label unique masses. With the modulo-four selection rule disabled, some outer target sectors retain two rivals; the finite-rival directional calculation gives $\widehat P_{\mathrm{USD}}=0.75$ and Bell-label masses $(1,1,0.5,0.5)$ for $(\psi^+,\psi^-,\phi^+,\phi^-)$. The independent permanent evaluation of all $1716$ records agrees with a maximum Bell-label mass difference below $9\times10^{-16}$.

At $M=16$, three TTN update paths reproduce the same aggregate, Bell-label-resolved, and root-sector masses, with the largest difference in $\widehat P_{\mathrm{USD}}$ equal to $1.33\times10^{-15}$. For the asymmetric instance, an $18$-point threshold sweep varied $\varepsilon_{\mathrm{supp}}$ from $10^{-14}$ to $10^{-10}$, $\varepsilon_{\mathrm{svd}}$ from $10^{-13}$ to $10^{-12}$, $\varepsilon_0$ from $10^{-11}$ to $10^{-9}$, and $\varepsilon_{\mathrm{merge}}$ from $10^{-8}$ to $10^{-6}$. The aggregate, Bell-label, and root-sector masses remain unchanged at the stored numerical precision, with all required decisions resolved.

Across the reported calculations, the largest recorded absolute discarded squared-singular-value weight is $2.71\times10^{-27}$, well below the $10^{-20}$ acceptance bound. State normalization, $0\le\widehat u_t(q_A)\le\widehat w_t(q_A)$, and probability-range conditions are also checked for every accepted calculation.

At $M=32$, five $\gamma$-block compositions were evaluated with both the TTN support calculation and a separately implemented half-circuit calculation. Across the aggregate and sign-resolved masses, the maximum absolute discrepancy is below $4.1\times10^{-10}$, within the comparison tolerance of $2\times10^{-7}$.

For every reported overlap sector, the retained-rival half-root Gram ranks required by Proposition~\ref{prop:directional_exactness} are checked before the directional assembly is accepted. The reported sectors also have rank-one target half-root support, which is checked because the numerical evaluation uses the scalar specialization in Eq.~\eqref{eq:one-rival-sector-mass}. A higher-rank target is allowed by Proposition~\ref{prop:directional_exactness}, with its exact matrix assembly given by Eq.~\eqref{eq:finite_rival_sector_mass}; a higher-rank rival invalidates the separated-half zero criterion unless additional joint left--right information or another structural condition excludes root cancellation.

Together, these checks provide finite-instance validation of the reported numerical results. They do not establish exact support beyond the analytical statements proved above.

\subsection{Numerical workflow and principal contraction costs}
\label{app:cost}

The numerical evaluation consists of TTN evolution, construction of the half-tree environment, directional transfer of the retained rivals, and root-sector assembly. If an applicability check fails or a required numerical decision remains unresolved, no aggregate result is reported. The bounds below describe the principal tensor contractions and are not a complete wall-time model of the numerical calculation.

Here $M$ is the number of optical modes; $r=r_{\mathrm{sec}}^{\mathrm{peak}}$ bounds the represented photon-number-sector dimensions entering these contractions; $S$ bounds the retained directional-record count at a tracked half-tree node; $n_{\mathrm{riv}}=|\mathcal K_t(q_A)|$ is the number of retained rival Bell labels; $a_{gq}\times b_{gq}$ is a nonempty reshaped photon-number block at refactorization step $g$ and sector $q$; and $B_{\mathrm{env}}$ counts compatible child-charge block-pair contractions in one half-tree environment construction.

At refactorization step $g$ and photon-number sector $q$, a nonempty reshaped block of size $a_{gq}\times b_{gq}$ is factorized with cost $O\!\left(a_{gq}b_{gq}\min(a_{gq},b_{gq})\right)$. The total refactorization contribution is obtained by summing this cost over the blocks encountered during the routed updates. For the environment construction, each compatible child-charge block pair requires a double-layer contraction bounded by $O(r^4)$. With $B_{\mathrm{env}}$ compatible block-pair contractions, the environment cost is $O(B_{\mathrm{env}}r^4)$. For the target transfer, at most $S^2$ pairs of child directional records are combined at each of $O(M)$ tree nodes. Each target Gram contraction is bounded by $O(r^4)$, giving $O(MS^2r^4)$. A retained rival carries a boundary vector rather than a Gram matrix, so each rival contraction is bounded by $O(r^3)$. Propagating $n_{\mathrm{riv}}$ rivals therefore gives $O(Mn_{\mathrm{riv}}S^2r^3)$. These estimates do not include the additional bookkeeping required for directional-record management and routing.

The finite-rival inclusion--exclusion has cost $O(3^{n_{\mathrm{riv}}})$. There are only four Bell inputs, so $n_{\mathrm{riv}}\le3$. In the overlap sectors used for the reported results, $n_{\mathrm{riv}}=1$, and this step is not a leading cost in the reported calculations.

Local gate application and routing depend on the actual charge-block shapes and routing sequence. Thus $\chi_{\mathrm{peak}}$ characterizes the size of the TTN representation but does not by itself determine the evolution time. These are component-wise bounds for the finite calculations considered here, not an asymptotic complexity theorem. No analytical $M$-scaling is established for the represented sector dimensions, bond dimensions, retained directional-record count, or routed refactorizations.

\section{Half-resource composition and reflection covariance}
\label{app:half-additivity}

\subsection{Composition identity for factorized half resources}
\label{app:half-additivity-derivation}

Let the normalized ancillary input factorize across the global analyzer cut as
\[
|\mathcal A_A\rangle\otimes|\mathcal A_B\rangle,
\]
where $|\mathcal A_X\rangle$ carries a definite photon number $Q_X$ on half $X\in\{A,B\}$. The complete input contains $Q=Q_A+Q_B+2$ photons. For a half occupation vector $\mathbf n_X$, define the half-restricted alternating residue using the global physical mode parity,
\[
\nu_X(\mathbf n_X)
=
\left[\sum_{j\in X}(-1)^j n_j\right]_4,
\qquad X\in\{A,B\}.
\]
Here, the sum runs over the physical modes belonging to half $X$. Each half resource is assumed to have a definite value of this residue on its nonzero Fock branches; the Bell slots are unoccupied in the ancillary state. Both halves are defined on the same half-mode space and are acted on by the fixed half analyzer $U_{\mathrm h}^{(N)}$.

For the resources used here, the elementary block $\gamma_1$ has residue $2$, while for $\tau\ge2$ the two branches of $\gamma_\tau$ have alternating differences $\pm2^\tau\equiv0\pmod4$. A recursive half therefore has residue $2$ because it contains exactly one $\gamma_1$, and a $\gamma_1$-product half has residue $2$ because it contains the odd number $2^N-1$ of elementary blocks. The mixed half $\gamma_2^{\otimes3}\otimes\gamma_1$, and its reflected ordering, also has residue $2$.

Let $s\in\{+1,-1\}$ denote the numerical Bell sign. On the two Bell modes of one analyzer half, define the normalized state
\[
|\theta_s\rangle
=
\frac{|2,0\rangle+s|0,2\rangle}{\sqrt2}.
\]
For fixed half resources, define
\[
\begin{aligned}
y_A(\mathbf n_A)
&=
\langle\mathbf n_A|\Gamma(U_{\mathrm h}^{(N)})
\bigl(|\mathcal A_A\rangle\otimes|0,0\rangle\bigr),\\
x_A^s(\mathbf n_A)
&=
\langle\mathbf n_A|\Gamma(U_{\mathrm h}^{(N)})
\bigl(|\mathcal A_A\rangle\otimes|\theta_s\rangle\bigr),
\end{aligned}
\]
whereas
\[
\begin{aligned}
y_B(\mathbf n_B)
&=
\langle\mathbf n_B|\Gamma(U_{\mathrm h}^{(N)})
\bigl(|0,0\rangle\otimes|\mathcal A_B\rangle\bigr),\\
x_B^s(\mathbf n_B)
&=
\langle\mathbf n_B|\Gamma(U_{\mathrm h}^{(N)})
\bigl(|\theta_s\rangle\otimes|\mathcal A_B\rangle\bigr).
\end{aligned}
\]
Thus, $y_X$ is defined on $\Omega_{X,Q_X}$ and $x_X^s$ on $\Omega_{X,Q_X+2}$, with
\[
\sum_{\mathbf n_X\in\Omega_{X,Q_X}}|y_X(\mathbf n_X)|^2
=
\sum_{\mathbf n_X\in\Omega_{X,Q_X+2}}|x_X^s(\mathbf n_X)|^2
=1.
\]
When used as an argument of $\kappa_X^s$, $\mathcal A_X$ denotes the half resource $|\mathcal A_X\rangle$. Define
\begin{equation}
\kappa_X^s(\mathcal A_X)
=
\sum_{\substack{
\mathbf n_X\in\Omega_{X,Q_X+2}\\
x_X^{-s}(\mathbf n_X)=0
}}
\left|x_X^s(\mathbf n_X)\right|^2,
\label{eq:app-kappa-definition}
\end{equation}
so that $0\le\kappa_X^s\le1$. A separate condition $x_X^s\neq0$ is unnecessary because a vanishing target amplitude contributes zero to Eq.~\eqref{eq:app-kappa-definition}.

\noindent\textit{Factorized-half composition.}
Under the assumptions above,
\begin{equation}
\begin{aligned}
\mu_{\psi^+}&=\mu_{\psi^-}=1,\\
\mu_{\phi^s}
&=
\frac12\left[
\kappa_A^s(\mathcal A_A)
+
\kappa_B^s(\mathcal A_B)
\right],
\qquad s\in\{+1,-1\}.
\end{aligned}
\label{eq:half-additivity-phi}
\end{equation}

\noindent\textit{Proof.}
Let the two ancillary half residues be $\nu_A$ and $\nu_B$. The complete ancillary residue is $[\nu_A+\nu_B]_4$; the Bell components of $\psi^\pm$ add residue $0$, whereas those of $\phi^\pm$ add residue $2$. Conservation of the residue therefore separates $\psi^+$ from $\phi^\pm$ in the outer sectors, while the central root-photon-number sector isolates $\psi^-$. Hence $\mu_{\psi^+}=\mu_{\psi^-}=1$. In the outer sector $q_A=Q_A+2$, both Bell photons enter half $A$, and Eq.~\eqref{eq:bell_layer_action} gives
\begin{equation}
c_{\phi^s}(\mathbf n_A,\mathbf n_B)
=
\frac{\mathrm i}{\sqrt2}\,
x_A^s(\mathbf n_A)y_B(\mathbf n_B).
\label{eq:app-outer-amplitude}
\end{equation}
For a nonzero target amplitude, the common factor $y_B$ is nonzero, so the opposite $\phi$ sign vanishes exactly when $x_A^{-s}=0$. The Eq.~\eqref{eq:app-kappa-definition} then gives
\[
u_{\phi^s}(Q_A+2)=\frac12\kappa_A^s(\mathcal A_A).
\]
Interchanging the two halves gives
\[
u_{\phi^s}(Q_A)=\frac12\kappa_B^s(\mathcal A_B).
\]
The two contributions belong to disjoint root-photon-number sectors and therefore add, yielding Eq.~\eqref{eq:half-additivity-phi}.
\hfill$\square$

With equal Bell priors,
\begin{equation}
P_{\mathrm{USD}}
=
\frac12
+
\frac18
\sum_{s=\pm1}
\left[
\kappa_A^s(\mathcal A_A)
+
\kappa_B^s(\mathcal A_B)
\right].
\label{eq:half-additivity-pusd}
\end{equation}
The Eq~\eqref{eq:half-additivity-pusd} follows from the Bell-layer sector decomposition and the two local half analyzers, independently of the TTN representation.

\subsection{Reflection covariance and asymmetric composition}

Let $R_{\mathrm h}$ reverse the physical mode order of one complete analyzer half. For a half resource $\mathcal A_X$, let $\mathcal A_X^{\mathrm R}$ denote the state with reversed ancillary-mode order and let $\mathbf n_X^{\mathrm R}$ denote the reversed detector pattern. Let $\mathsf S_2$ denote the two-mode swap. In the binary physical-mode labeling used for $U_{\mathrm h}^{(N)}$, complete reversal is
\[
R_{\mathrm h}=\mathsf S_2^{\otimes(N+1)},
\qquad
[R_{\mathrm h},U_{\mathrm h}^{(N)}]=0,
\]
because the balanced beam splitter commutes with the two-mode swap $\mathsf S_2$, since a half contains an even number of modes, reversal exchanges even and odd positions, so the alternating residue changes sign modulo four; in particular, a definite-residue half remains definite residue after reflection.

For an arbitrary state $|\chi\rangle$ of the two Bell slots, Fock-space reversal gives
\[
\Gamma(R_{\mathrm h})
\bigl(|\mathcal A_A\rangle\otimes|\chi\rangle\bigr)
=
(\mathsf S_2|\chi\rangle)\otimes|\mathcal A_A^{\mathrm R}\rangle.
\]
For the local Bell-mode states defined above,
\[
\mathsf S_2|0,0\rangle=|0,0\rangle,
\qquad
\mathsf S_2|\theta_s\rangle=s|\theta_s\rangle.
\]
Restoring the resource as an explicit argument when comparing different half inputs, the commutation relation above gives
\[
\begin{aligned}
y_B(\mathbf n_A^{\mathrm R};\mathcal A_A^{\mathrm R})
&=y_A(\mathbf n_A;\mathcal A_A),\\
x_B^s(\mathbf n_A^{\mathrm R};\mathcal A_A^{\mathrm R})
&=s\,x_A^s(\mathbf n_A;\mathcal A_A).
\end{aligned}
\]
The factor $s$ changes neither support zeros nor squared amplitudes, and record reversal is one-to-one. Hence
\begin{equation}
\kappa_A^s(\mathcal A_A)
=
\kappa_B^s(\mathcal A_A^{\mathrm R}).
\label{eq:app-reflection-kappa}
\end{equation}
Reflection therefore preserves the sign-resolved half-resource fraction and does not exchange the two $\phi$ signs.

\noindent\textit{Reflected-endpoint corollary.}
Let $\mathcal A_1$ and $\mathcal A_2$ be normalized half resources on the same half-mode space, acted on by the same fixed half analyzer, with definite photon number and definite alternating residue. With equal Bell priors, let $P_{\mathrm{sym}}(\mathcal A_j)$ denote the success probability for $(\mathcal A_j,\mathcal A_j^{\mathrm R})$ and $P_{\mathrm{asym}}(\mathcal A_1,\mathcal A_2)$ and that for $(\mathcal A_1,\mathcal A_2^{\mathrm R})$. Equations~\eqref{eq:half-additivity-pusd} and~\eqref{eq:app-reflection-kappa} then give
\begin{equation}
P_{\mathrm{asym}}(\mathcal A_1,\mathcal A_2)
=
\frac12\left[
P_{\mathrm{sym}}(\mathcal A_1)
+
P_{\mathrm{sym}}(\mathcal A_2)
\right].
\label{eq:app-reflection-average}
\end{equation}
The mean comes from the two disjoint outer root-photon-number sectors of the same coherent Bell-plus-ancilla input, not from a classical mixture of endpoint experiments. Under the assumptions of the corollary, the asymmetric pairing cannot outperform the better-reflected symmetric endpoint.

\noindent\textit{Elementary-block endpoint.}
For the $M=32$ $\gamma_1$-product endpoint, write $F_2=2^{-1/2}\bigl(\begin{smallmatrix}1&1\\1&-1\end{smallmatrix}\bigr)$ and $D=\operatorname{diag}(1,\mathrm i)$. Since $U_{\mathrm{BS}}=DF_2D$, a fixed permutation that groups the rail index first maps the half analyzer to $\mathbb I_2\otimes F_2^{\otimes N}$ up to diagonal input/output phases. Output phases do not change photon-counting probabilities, and the permutation only relabels detector records. In the outer root sector, each two-rail position has a total occupation of 0 or 2, so the remaining input-position phases are common and do not change the support. A $\pi/2$ phase on the second rail maps each $|\gamma_1\rangle$ block to $(|2,0\rangle-|0,2\rangle)/\sqrt2$ and only exchanges the local $\phi$ signs.

Yamazaki \textit{et al.}~\cite{YamazakiIkutaYamamoto2023} also consider the tensor-product $F_2^{\otimes N}$ replacement. Their Supplemental Lemma~S4 gives the group condition required for the same probability calculation. The binary translation group $(\mathbb Z_2)^N$ satisfies this condition because, for any two of the $2^N$ positions, exactly one translation maps one to the other. Corollary~9 of the cited arXiv version therefore applies to the present endpoint. For $N=3$, corresponding to eight positions, it gives $P_{\mathrm{sym}}=403/512$.

For the reflected symmetric endpoints, the recursive and $\gamma_1$-product success probabilities are $7/8$ and $25/32$ for $M=16$, and $15/16$ and $403/512$ for $M=32$~\cite{EwertvanLoock2014,YamazakiIkutaYamamoto2023}. Equation~\eqref{eq:app-reflection-average} therefore gives
\begin{equation}
\begin{aligned}
P_{\mathrm{asym}}(M=16)
&=
\frac12\left(\frac78+\frac{25}{32}\right)
=\frac{53}{64},\\
P_{\mathrm{asym}}(M=32)
&=
\frac12\left(\frac{15}{16}+\frac{403}{512}\right)
=\frac{883}{1024}.
\end{aligned}
\label{eq:app-asymmetric-success-values}
\end{equation}
The decomposition above identifies the corresponding label- and sector-resolved contributions. Their numerical values reported here remain finite-threshold TTN quantities and are independently checked by the half-circuit calculation.

\bibliography{ref}

@article{Calsamiglia2001,
  title={Maximum efficiency of a linear-optical {Bell}-state analyzer},
  author={Calsamiglia, John and L{\"u}tkenhaus, Norbert},
  journal={Applied Physics B},
  volume={72},
  number={1},
  pages={67--71},
  year={2001},
  publisher={Springer},
  doi={10.1007/s003400000484},
  url={https://doi.org/10.1007/s003400000484}}

@article{Lutkenhaus1999,
  title = {{Bell} measurements for teleportation},
  author = {L\"utkenhaus, N. and Calsamiglia, J. and Suominen, K.-A.},
  journal = {Phys. Rev. A},
  volume = {59},
  issue = {5},
  pages = {3295--3300},
  numpages = {0},
  year = {1999},
  month = {May},
  publisher = {American Physical Society},
  doi = {10.1103/PhysRevA.59.3295},
  url = {https://link.aps.org/doi/10.1103/PhysRevA.59.3295}}

@article{Grice2011,
  title = {Arbitrarily complete {Bell}-state measurement using only linear optical elements},
  author = {Grice, W. P.},
  journal = {Phys. Rev. A},
  volume = {84},
  issue = {4},
  pages = {042331},
  numpages = {6},
  year = {2011},
  month = {Oct},
  publisher = {American Physical Society},
  doi = {10.1103/PhysRevA.84.042331},
  url = {https://link.aps.org/doi/10.1103/PhysRevA.84.042331}}

@article{EwertvanLoock2014,
  title = {$3/4$-Efficient {Bell} Measurement with Passive Linear Optics and Unentangled Ancillae},
  author = {Ewert, Fabian and van Loock, Peter},
  journal = {Phys. Rev. Lett.},
  volume = {113},
  issue = {14},
  pages = {140403},
  numpages = {5},
  year = {2014},
  month = {Sep},
  publisher = {American Physical Society},
  doi = {10.1103/PhysRevLett.113.140403},
  url = {https://link.aps.org/doi/10.1103/PhysRevLett.113.140403}}

@article{OlivoGrosshans2018,
  title = {Ancilla-assisted linear optical {Bell} measurements and their optimality},
  author = {Olivo, Andrea and Grosshans, Fr\'ed\'eric},
  journal = {Phys. Rev. A},
  volume = {98},
  issue = {4},
  pages = {042323},
  numpages = {10},
  year = {2018},
  month = {Oct},
  publisher = {American Physical Society},
  doi = {10.1103/PhysRevA.98.042323},
  url = {https://link.aps.org/doi/10.1103/PhysRevA.98.042323}}

@article{Wein2016,
  title = {Efficiency of an enhanced linear optical {Bell}-state measurement scheme with realistic imperfections},
  author = {Wein, Stephen and Heshami, Khabat and Fuchs, Christopher A. and Krovi, Hari and Dutton, Zachary and Tittel, Wolfgang and Simon, Christoph},
  journal = {Phys. Rev. A},
  volume = {94},
  issue = {3},
  pages = {032332},
  numpages = {13},
  year = {2016},
  month = {Sep},
  publisher = {American Physical Society},
  doi = {10.1103/PhysRevA.94.032332},
  url = {https://link.aps.org/doi/10.1103/PhysRevA.94.032332}}

@article{Chefles1998,
title = {Unambiguous discrimination between linearly independent quantum states},
journal = {Physics Letters A},
volume = {239},
number = {6},
pages = {339-347},
year = {1998},
issn = {0375-9601},
doi = {10.1016/S0375-9601(98)00064-4},
url = {https://www.sciencedirect.com/science/article/pii/S0375960198000644},
author = {Anthony Chefles}}

@article{SinghPfeiferVidal2010,
  title = {Tensor network decompositions in the presence of a global symmetry},
  author = {Singh, Sukhwinder and Pfeifer, Robert N. C. and Vidal, Guifr\'e},
  journal = {Phys. Rev. A},
  volume = {82},
  issue = {5},
  pages = {050301(R)},
  numpages = {4},
  year = {2010},
  month = {Nov},
  publisher = {American Physical Society},
  doi = {10.1103/PhysRevA.82.050301},
  url = {https://link.aps.org/doi/10.1103/PhysRevA.82.050301}}

@article{SinghPfeiferVidal2011,
  title = {Tensor network states and algorithms in the presence of a global {U(1)} symmetry},
  author = {Singh, Sukhwinder and Pfeifer, Robert N. C. and Vidal, Guifre},
  journal = {Phys. Rev. B},
  volume = {83},
  issue = {11},
  pages = {115125},
  numpages = {22},
  year = {2011},
  month = {Mar},
  publisher = {American Physical Society},
  doi = {10.1103/PhysRevB.83.115125},
  url = {https://link.aps.org/doi/10.1103/PhysRevB.83.115125}}

@article{ShiDuanVidal2006,
  title = {Classical simulation of quantum many-body systems with a tree tensor network},
  author = {Shi, Y.-Y. and Duan, L.-M. and Vidal, G.},
  journal = {Phys. Rev. A},
  volume = {74},
  issue = {2},
  pages = {022320},
  numpages = {4},
  year = {2006},
  month = {Aug},
  publisher = {American Physical Society},
  doi = {10.1103/PhysRevA.74.022320},
  url = {https://link.aps.org/doi/10.1103/PhysRevA.74.022320}}

@article{Orus2014,
title = {A practical introduction to tensor networks: Matrix product states and projected entangled pair states},
journal = {Annals of Physics},
volume = {349},
pages = {117-158},
year = {2014},
issn = {0003-4916},
doi = {10.1016/j.aop.2014.06.013},
url = {https://www.sciencedirect.com/science/article/pii/S0003491614001596},
author = {Román Orús}}

@article{Oh2024,
  title={Classical algorithm for simulating experimental {Gaussian} boson sampling},
  author={Oh, Changhun and Liu, Minzhao and Alexeev, Yuri and Fefferman, Bill and Jiang, Liang},
  journal={Nature Physics},
  volume={20},
  number={9},
  pages={1461--1468},
  year={2024},
  publisher={Nature Publishing Group UK London},
  doi={10.1038/s41567-024-02535-8},
  url={https://doi.org/10.1038/s41567-024-02535-8}}

@article{OhGraph2022,
  title = {Classical Simulation of Boson Sampling Based on Graph Structure},
  author = {Oh, Changhun and Lim, Youngrong and Fefferman, Bill and Jiang, Liang},
  journal = {Phys. Rev. Lett.},
  volume = {128},
  issue = {19},
  pages = {190501},
  numpages = {7},
  year = {2022},
  month = {May},
  publisher = {American Physical Society},
  doi = {10.1103/PhysRevLett.128.190501},
  url = {https://link.aps.org/doi/10.1103/PhysRevLett.128.190501}}

@article{Cilluffo2026,
  title = {Heisenberg picture tensor network formalism for optical circuits},
  author = {Cilluffo, Dario and Kost, Matthias and Lorenzoni, Nicola and Plenio, Martin B.},
  journal = {Phys. Rev. Res.},
  volume = {8},
  issue = {2},
  pages = {023098},
  numpages = {9},
  year = {2026},
  month = {Apr},
  publisher = {American Physical Society},
  doi = {10.1103/wtwl-979d},
  url = {https://link.aps.org/doi/10.1103/wtwl-979d}}

@article{VintherKastoryano2025,
  title = {Variational tensor network simulation of {Gaussian} boson sampling and beyond},
  author = {Vinther, Jonas and Kastoryano, Michael J.},
  journal = {Phys. Rev. A},
  volume = {112},
  issue = {2},
  pages = {022605},
  numpages = {12},
  year = {2025},
  month = {Aug},
  publisher = {American Physical Society},
  doi = {10.1103/z463-7gqy},
  url = {https://link.aps.org/doi/10.1103/z463-7gqy}}

@article{Shchesnovich2015,
  title = {Partial indistinguishability theory for multiphoton experiments in multiport devices},
  author = {Shchesnovich, V. S.},
  journal = {Phys. Rev. A},
  volume = {91},
  issue = {1},
  pages = {013844},
  numpages = {16},
  year = {2015},
  month = {Jan},
  publisher = {American Physical Society},
  doi = {10.1103/PhysRevA.91.013844},
  url = {https://link.aps.org/doi/10.1103/PhysRevA.91.013844}}

@article{BennettTeleportation1993,
  title = {Teleporting an unknown quantum state via dual classical and Einstein-Podolsky-Rosen channels},
  author = {Bennett, Charles H. and Brassard, Gilles and Cr\'epeau, Claude and Jozsa, Richard and Peres, Asher and Wootters, William K.},
  journal = {Phys. Rev. Lett.},
  volume = {70},
  issue = {13},
  pages = {1895--1899},
  numpages = {0},
  year = {1993},
  month = {Mar},
  publisher = {American Physical Society},
  doi = {10.1103/PhysRevLett.70.1895},
  url = {https://link.aps.org/doi/10.1103/PhysRevLett.70.1895}}

@article{ZukowskiEntanglementSwapping1993,
  title = {``Event-ready-detectors'' {Bell} experiment via entanglement swapping},
  author = {\ifmmode \dot{Z}\else \.{Z}\fi{}ukowski, M. and Zeilinger, A. and Horne, M. A. and Ekert, A. K.},
  journal = {Phys. Rev. Lett.},
  volume = {71},
  issue = {26},
  pages = {4287--4290},
  numpages = {0},
  year = {1993},
  month = {Dec},
  publisher = {American Physical Society},
  doi = {10.1103/PhysRevLett.71.4287},
  url = {https://link.aps.org/doi/10.1103/PhysRevLett.71.4287}}

@article{KLM2001,
  title={A scheme for efficient quantum computation with linear optics},
  author={Knill, Emanuel and Laflamme, Raymond and Milburn, Gerald J},
  journal={Nature},
  volume={409},
  number={6816},
  pages={46--52},
  year={2001},
  publisher={Nature Publishing Group},
  doi={10.1038/35051009},
  url={https://doi.org/10.1038/35051009}}

@article{BrowneRudolph2005,
  title = {Resource-Efficient Linear Optical Quantum Computation},
  author = {Browne, Daniel E. and Rudolph, Terry},
  journal = {Phys. Rev. Lett.},
  volume = {95},
  issue = {1},
  pages = {010501},
  numpages = {4},
  year = {2005},
  month = {Jun},
  publisher = {American Physical Society},
  doi = {10.1103/PhysRevLett.95.010501},
  url = {https://link.aps.org/doi/10.1103/PhysRevLett.95.010501}}

@article{KokRMP2007,
  title = {Linear optical quantum computing with photonic qubits},
  author = {Kok, Pieter and Munro, W. J. and Nemoto, Kae and Ralph, T. C. and Dowling, Jonathan P. and Milburn, G. J.},
  journal = {Rev. Mod. Phys.},
  volume = {79},
  issue = {1},
  pages = {135--174},
  numpages = {0},
  year = {2007},
  month = {Jan},
  publisher = {American Physical Society},
  doi = {10.1103/RevModPhys.79.135},
  url = {https://link.aps.org/doi/10.1103/RevModPhys.79.135}}

@article{BianchiMarconiBacco2026,
  title={{Bell} state measurements in quantum optics: a review of recent progress and open challenges},
  author={Bianchi, Luca and Marconi, Carlo and Bacco, Davide},
  journal={Quantum Science and Technology},
  volume={11},
  number={2},
  pages={023001},
  year={2026},
  publisher={IOP Publishing},
  doi={10.1088/2058-9565/ae609e},
  url={https://doi.org/10.1088/2058-9565/ae609e}}

@article{Bartolucci2023FBQC,
  title={Fusion-based quantum computation},
  author={Bartolucci, Sara and Birchall, Patrick and Bomb{\'i}n, Hector and Cable, Hugo and Dawson, Chris and Gimeno-Segovia, Mercedes and Johnston, Eric and Kieling, Konrad and Nickerson, Naomi and Pant, Mihir and Pastawski, Fernando and Rudolph, Terry and Sparrow, Chris},
  journal={Nature Communications},
  volume={14},
  pages={912},
  year={2023},
  publisher={Nature Publishing Group},
  doi={10.1038/s41467-023-36493-1},
  url={https://doi.org/10.1038/s41467-023-36493-1}}

@article{DAurelio2025BoostedTeleportation,
  title={Boosted quantum teleportation},
  author={D'Aurelio, Simone E. and Bayerbach, Matthias J. and Barz, Stefanie},
  journal={npj Quantum Information},
  volume={11},
  pages={37},
  year={2025},
  publisher={Nature Publishing Group},
  doi={10.1038/s41534-025-00992-4},
  url={https://doi.org/10.1038/s41534-025-00992-4}}

@article{Hauser2025BoostedBSM,
  title={Boosted {Bell}-state measurements for photonic quantum computation},
  author={Hauser, Nico and Bayerbach, Matthias J and D’Aurelio, Simone E and Weber, Raphael and Santandrea, Matteo and Kumar, Shreya P and Dhand, Ish and Barz, Stefanie},
  journal={npj Quantum Information},
  volume={11},
  number={1},
  pages={41},
  year={2025},
  publisher={Nature Publishing Group},
  doi={10.1038/s41534-025-00986-2},
  url={https://doi.org/10.1038/s41534-025-00986-2}}

@article{Asenbeck2024HybridBSM,
  title = {Hybrid Approach to Mitigate Errors in Linear Photonic {Bell}-State Measurement for Quantum Interconnects},
  author = {Asenbeck, Beate E. and Kawasaki, Akito and Boyer, Ambroise and Darras, Tom and Urvoy, Alban and Furusawa, Akira and Laurat, Julien},
  journal = {PRX Quantum},
  volume = {5},
  issue = {3},
  pages = {030331},
  numpages = {10},
  year = {2024},
  month = {Aug},
  publisher = {American Physical Society},
  doi = {10.1103/PRXQuantum.5.030331},
  url = {https://link.aps.org/doi/10.1103/PRXQuantum.5.030331}}

@article{Bayerbach2023,
  title={{Bell}-state measurement exceeding 50\% success probability with linear optics},
  author={Bayerbach, Matthias J and D’Aurelio, Simone E and van Loock, Peter and Barz, Stefanie},
  journal={Science Advances},
  volume={9},
  number={32},
  pages={eadf4080},
  year={2023},
  publisher={American Association for the Advancement of Science},
  doi={10.1126/sciadv.adf4080},
  url={https://doi.org/10.1126/sciadv.adf4080}}

@article{ZaidiVanLoock2013,
  title = {Beating the One-Half Limit of Ancilla-Free Linear Optics {Bell} Measurements},
  author = {Zaidi, Hussain A. and van Loock, Peter},
  journal = {Phys. Rev. Lett.},
  volume = {110},
  issue = {26},
  pages = {260501},
  numpages = {5},
  year = {2013},
  month = {Jun},
  publisher = {American Physical Society},
  doi = {10.1103/PhysRevLett.110.260501},
  url = {https://link.aps.org/doi/10.1103/PhysRevLett.110.260501}}

@misc{LeeJeong2013,
      title={{Bell}-state measurement and quantum teleportation using linear optics: two-photon pairs, entangled coherent states, and hybrid entanglement}, 
      author={Seung-Woo Lee and Hyunseok Jeong},
      year={2013},
      eprint={1304.1214},
      archivePrefix={arXiv},
      primaryClass={quant-ph},
      url={https://arxiv.org/abs/1304.1214}}

@article{LeeParkRalphJeong2015,
  title = {Nearly deterministic {Bell} measurement with multiphoton entanglement for efficient quantum-information processing},
  author = {Lee, Seung-Woo and Park, Kimin and Ralph, Timothy C. and Jeong, Hyunseok},
  journal = {Phys. Rev. A},
  volume = {92},
  issue = {5},
  pages = {052324},
  numpages = {9},
  year = {2015},
  month = {Nov},
  publisher = {American Physical Society},
  doi = {10.1103/PhysRevA.92.052324},
  url = {https://link.aps.org/doi/10.1103/PhysRevA.92.052324}}

@article{WeiBarreiroKwiat2007,
  title = {Hyperentangled {Bell}-state analysis},
  author = {Wei, Tzu-Chieh and Barreiro, Julio T. and Kwiat, Paul G.},
  journal = {Phys. Rev. A},
  volume = {75},
  issue = {6},
  pages = {060305(R)},
  numpages = {4},
  year = {2007},
  month = {Jun},
  publisher = {American Physical Society},
  doi = {10.1103/PhysRevA.75.060305},
  url = {https://link.aps.org/doi/10.1103/PhysRevA.75.060305}}

@article{PisentiGaeblerLynn2011,
  title = {Distinguishability of hyperentangled {Bell} states by linear evolution and local projective measurement},
  author = {Pisenti, N. and Gaebler, C. P. E. and Lynn, T. W.},
  journal = {Phys. Rev. A},
  volume = {84},
  issue = {2},
  pages = {022340},
  numpages = {5},
  year = {2011},
  month = {Aug},
  publisher = {American Physical Society},
  doi = {10.1103/PhysRevA.84.022340},
  url = {https://link.aps.org/doi/10.1103/PhysRevA.84.022340}}

@article{AkinNonlinear2025,
  title = {Faithful Quantum Teleportation via a Nanophotonic Nonlinear {Bell} State Analyzer},
  author = {Akin, Joshua and Zhao, Yunlei and Kwiat, Paul G. and Goldschmidt, Elizabeth A. and Fang, Kejie},
  journal = {Phys. Rev. Lett.},
  volume = {134},
  issue = {16},
  pages = {160802},
  numpages = {8},
  year = {2025},
  month = {Apr},
  publisher = {American Physical Society},
  doi = {10.1103/PhysRevLett.134.160802},
  url = {https://link.aps.org/doi/10.1103/PhysRevLett.134.160802}}

@misc{YamazakiIkutaYamamoto2023,
      title={Stabilizer formalism in linear optics and application to {Bell}-state discrimination}, 
      author={Tomohiro Yamazaki and Rikizo Ikuta and Takashi Yamamoto},
      year={2023},
      eprint={2301.06551},
      archivePrefix={arXiv},
      primaryClass={quant-ph},
      url={https://arxiv.org/abs/2301.06551}}

@misc{LahaVanLoock2026,
  title = {Auxiliary {Schmidt} Rank as a Resource for Photonic {Bell} Measurements},
  author = {Laha, Pradip and van Loock, Peter},
  year = {2026},
  eprint = {2606.24591},
  archivePrefix = {arXiv},
  primaryClass = {quant-ph},
  url = {https://arxiv.org/abs/2606.24591}
}

@article{Hilaire2023,
  title = {Linear Optical Logical {Bell} State Measurements with Optimal Loss-Tolerance Threshold},
  author = {Hilaire, Paul and Castor, Yaron and Barnes, Edwin and Economou, Sophia E. and Grosshans, Fr\'ed\'eric},
  journal = {PRX Quantum},
  volume = {4},
  issue = {4},
  pages = {040322},
  numpages = {18},
  year = {2023},
  month = {Nov},
  publisher = {American Physical Society},
  doi = {10.1103/PRXQuantum.4.040322},
  url = {https://link.aps.org/doi/10.1103/PRXQuantum.4.040322}}

@article{ReissVanLoock2026,
  title = {Optimal logical Bell measurements on stabilizer codes with linear optics},
  author = {Reiß, Simon D. and Loock, Peter van},
  journal = {Phys. Rev. A},
  pages = {},
  year = {2026},
  month = {Aug},
  publisher = {American Physical Society},
  doi = {10.1103/qj8s-j274},
  url = {https://link.aps.org/doi/10.1103/qj8s-j274}
}

@inproceedings{CliffordClifford2018,
  title={The classical complexity of boson sampling},
  author={Clifford, Peter and Clifford, Rapha{\"e}l},
  booktitle={Proceedings of the Twenty-Ninth Annual ACM-SIAM Symposium on Discrete Algorithms},
  pages={146--155},
  year={2018},
  organization={SIAM},
  doi={10.1137/1.9781611975031.10},
  url={https://doi.org/10.1137/1.9781611975031.10}}

@misc{Scheel2008,
  author = {Scheel, Stefan},
  title = {Permanents in linear optical networks},
  year = {2004},
  eprint = {quant-ph/0406127},
  archivePrefix = {arXiv},
  primaryClass = {quant-ph},
  url = {https://arxiv.org/abs/quant-ph/0406127}
}

@misc{SmithKaplan2018,
      title={Approaching near-perfect state discrimination of photonic Bell states through the use of unentangled ancilla photons}, 
      author={Jake A. Smith and Lev Kaplan},
      year={2018},
      eprint={1802.10527},
      archivePrefix={arXiv},
      primaryClass={quant-ph},
      url={https://arxiv.org/abs/1802.10527}, 
}

@article{Ivanovic1987,
title = {How to differentiate between non-orthogonal states},
journal = {Physics Letters A},
volume = {123},
number = {6},
pages = {257-259},
year = {1987},
issn = {0375-9601},
doi = {10.1016/0375-9601(87)90222-2},
url = {https://www.sciencedirect.com/science/article/pii/0375960187902222},
author = {I.D. Ivanovic}}

@article{Dieks1988,
title = {Overlap and distinguishability of quantum states},
journal = {Physics Letters A},
volume = {126},
number = {5},
pages = {303-306},
year = {1988},
issn = {0375-9601},
doi = {10.1016/0375-9601(88)90840-7},
url = {https://www.sciencedirect.com/science/article/pii/0375960188908407},
author = {D. Dieks}}

@article{Peres1988,
title = {How to differentiate between non-orthogonal states},
journal = {Physics Letters A},
volume = {128},
number = {1},
pages = {19},
year = {1988},
issn = {0375-9601},
doi = {10.1016/0375-9601(88)91034-1},
url = {https://www.sciencedirect.com/science/article/pii/0375960188910341},
author = {Asher Peres}}

@article{roga2020classical,
  title={Classical simulation of boson sampling with sparse output},
  author={Roga, Wojciech and Takeoka, Masahiro},
  journal={Scientific Reports},
  volume={10},
  number={1},
  pages={14739},
  year={2020},
  publisher={Nature Publishing Group UK London},
  doi={10.1038/s41598-020-71892-0},
  url={https://doi.org/10.1038/s41598-020-71892-0}}

@article{jacob2020franck,
  title = {{Franck-Condon} factors via compressive sensing},
  author = {Jacob, Kevin Valson and Kaur, Eneet and Roga, Wojciech and Takeoka, Masahiro},
  journal = {Phys. Rev. A},
  volume = {102},
  issue = {3},
  pages = {032403},
  numpages = {10},
  year = {2020},
  month = {Sep},
  publisher = {American Physical Society},
  doi = {10.1103/PhysRevA.102.032403},
  url = {https://link.aps.org/doi/10.1103/PhysRevA.102.032403}}

@article{tagliacozzo2009simulation,
  title = {Simulation of two-dimensional quantum systems using a tree tensor network that exploits the entropic area law},
  author = {Tagliacozzo, Luca and Evenbly, Glen and Vidal, Guifr{\'e}},
  journal = {Physical Review B},
  volume = {80},
  number = {23},
  pages = {235127},
  year = {2009},
  publisher = {American Physical Society},
  doi = {10.1103/PhysRevB.80.235127}}

@article{Seitz2023simulatingquantum,
  doi = {10.22331/q-2023-03-30-964},
  url = {https://doi.org/10.22331/q-2023-03-30-964},
  title = {Simulating quantum circuits using tree tensor networks},
  author = {Seitz, Philipp and Medina, Ismael and Cruz, Esther and Huang, Qunsheng and Mendl, Christian B.},
  journal = {Quantum},
  volume = {7},
  pages = {964},
  year = {2023}
}

\end{document}